\documentclass[3p,12pt]{elsarticle}
\usepackage{a4wide}
\usepackage{lineno}
\usepackage{graphicx}
\usepackage{tikz}
\usepackage{xcolor}
\usepackage{xspace}
\usetikzlibrary{calc,arrows.meta}
\usetikzlibrary{backgrounds}
\usetikzlibrary{intersections,patterns,angles,quotes,arrows,positioning}
\usepackage{amsfonts}
\usepackage{amsmath}
\usepackage{amsthm}
\usepackage{a4}
\usepackage{algpseudocode}
\usepackage{algorithm}
\usepackage{cases}
\usepackage{verbatim}
\usepackage{caption}
\usepackage{subcaption}
\DeclareCaptionLabelFormat{withfigure}{(\arabic{figure}#2)}
\biboptions{authoryear}

\modulolinenumbers[1]

\renewcommand{\phi}{\varphi}

\newcommand{\MA}{\text{MA}}
\newcommand{\MAT}{\text{MAT}}

\newtheorem{lemma}{Lemma}
\newtheorem{proposition}[lemma]{Proposition}

\newtheorem{theorem}[lemma]{Theorem}

\newtheorem*{lemma*}{Lemma}
\theoremstyle{definition}
\newtheorem{definition}{Definition}
\theoremstyle{remark}
\newtheorem*{remark}{Remark}
\renewcommand{\epsilon}{\varepsilon}
\newtheorem{example}{Example}
\newcommand{\f}[1]{\mathbf{#1}}
\newlength\vecx
\newlength\vecy
\newlength\dist

\begin{document}

\begin{frontmatter}

\title{\LARGE Arc Spline Approximation of Envelopes of Evolving Planar Domains}

\author{Jana Vr\'abl\'{\i}kov\'a$^*$}
\ead{jana.vrablikova@jku.at}
\author{Bert J\"uttler}
\ead{bert.juettler@jku.at}
\cortext[cor]{Corresponding author}
\address{Institute of Applied Geometry, Johannes Kepler University, Linz/Austria}

\begin{abstract}
\par Computing the envelope of deforming planar domains is a significant and challenging problem with a wide range of potential applications. We approximate the envelope using circular arc splines, curves that balance geometric flexibility and computational simplicity. Our approach combines two concepts to achieve these benefits.

First, we represent a planar domain by its medial axis transform (MAT), which is a~geometric graph in Minkowski space $\mathbb R^{2,1}$ (possibly with degenerate branches). We observe that circular arcs in the Minkowski space correspond to MATs of arc spline domains. Furthermore, as a planar domain evolves over time, each branch of its MAT evolves and forms a surface in the Minkowski space. This allows us to reformulate the problem of envelope computation as a problem of computing cyclographic images of finite sets of curves on these surfaces. We propose and compare two pairs of methods for approximating the curves and boundaries of their cyclographic images. All of these methods result in an arc spline approximation of the envelope of the evolving domain.

Second, we exploit the geometric flexibility of circular arcs in both the plane and Minkowski space to achieve a high approximation rate. The computational simplicity ensures the efficient trimming of redundant branches of the generated envelope using a~sweep line algorithm with optimal computational complexity.
\end{abstract}


\end{frontmatter}

\section{Introduction}
Sweeping a planar domain along a predefined trajectory is a powerful tool for e.g. constructing more complex domains and the main challenge is to approximate the envelope of the resulting volume or area. The case where the domain moves under a rigid body motion has been well studied, see e.g. \cite{Machchhar2022} and the references therein. 

Swept volumes can be generalized if we allow the domain to change its size or shape as it moves. We call such domains \emph{evolving domains}. In the literature, the process is sometimes called general sweep.

\subsection{Related work}
Some of the existing methods build upon the techniques for computation of the envelope of a domain undergoing a rigid body motion. For example, Sweep Differential Equation (SDE) and Sweep Envelope Differential Equation (SEDE) were proposed \cite{Blackmore1997,Blackmore1990}. They identify the sweeps with first--order linear ordinary equations. The methods were later generalized to general sweeps \cite{Blackmore1994, Wang2000}. 

Other methods develop the theory for evolving domains separately. Kim and Elber formulate the problem as a polynomial equation in three variables \cite{Kim2000}. Despite considering exact geometries, the technique generates a polyhedral approximation of the surface of the evolving domain. 

In the case of planar evolving domains, a method for obtaining the exact envelope was proposed in \cite{kim1993} for domains that move along a parametric curve trajectory while they evolve, i.e., change their shape parametrically. The algorithm generates curve segments from which the envelope is extracted by a plane sweep algorithm. The plane sweep algorithm proposed by Bentley and Ottmann \cite{Bentley1979} was designed to compute and report intersections in a set of line segments and can be generalized to other primitives. However, for segments of algebraic curves of high degree, the plane sweep algorithm is highly inefficient. 

To overcome the time difficulties, a polygonal approximation of the envelope of an evolving domain was proposed \cite{Ahn1993}. The envelope is again extracted by the plane sweep algorithm from a set of lines generated by the method. The polygonal boundaries are then fitted by cubic splines. 

An incremental algorithm was proposed by \cite{Lee2000}. At each instance of the time parameter, the domain is approximated by a polygonal boundary and the envelope is generated using boolean operations on simple polygons, which can be implemented robustly. The incremental nature of the algorithm produces the envelope at each step, such that no additional removal of redundant parts is necessary, making it useful for interactive shape design.

Our method is based on the interpolation of the medial axis transform. Medial axis (MA) was introduced by Blum \cite{Blum1967} as a tool for shape recognition. It describes a planar domain by the centres of medial disks, i.e., maximum disks that are inscribed in the domain and are tangent to its boundary at at least two points. The medial axis transform (MAT) is defined by the medial axis and the radius function of the medial disks. Besides shape recognition, the medial axis transform has a wide range of applications, for example in domain decomposition, where it can be then used for $G^1$-smooth parameterization of complex domains \cite{Pan2023}.

The medial axis transform can be computed exactly for domains with piecewise linear or piecewise circular boundaries \cite{Aichholzer2010,Chin1995,Evans1998,Lee1982}. For free--form shapes, the MAT can be approximated using a variety of techniques, including tracing methods  \cite{Cao2008,Degen2004,Ramanathan2003}, divide-and-conquer methods  \cite{Aichholzer2009,Choi1997}, and computing the Voronoi diagrams of sampled points \cite{Attali1997, Fabbri2002, Zhu2014}.

In \cite{Wolter92} it is shown that the medial axis of a compact set $\Omega$ with a piecewise $\mathcal C^2$ smooth boundary $\delta\Omega$ is path-connected. If $\delta\Omega$ consists of finitely many real analytic curves, in \cite{Choi97} it is shown that $\MA(\Omega)$ and $\MAT(\Omega)$ are connected geometric graphs with finitely many branches (edges) and vertices.

The branches of the medial axis transform can be seen as curves in
Minkowski space $\mathbb R^{2,1}$
\cite{Degen2004,Pottmann1998}. Minkowski space is a model of Laguerre
geometry, the Euclidean geometry of oriented spheres and
hyperplanes. Laguerre geometry was first introduced by Blaschke
\cite{blaschke1929} and reexamined in a modern manner by Cecil
\cite{Cecil1992}. Minkowski space (also called the cyclographic model)
of Laguerre geometry was described in \cite{Pottmann1998}. The model
estabilishes a one--to--one correspondence between oriented spheres
and points in Minkowski space. The mapping that assigns an oriented
circle to a point in Minkowski space is called the cyclographic
mapping. A cyclographic image of a point of the MAT of a planar domain
is a medial disc and the domain can be recreated as the envelope of
the medial discs.

The particular class of Minkowski Pythagorean-hodograph
(MPH) curves \cite{Moon1999}, which are rationally parameterized
curves that represent branches of the MAT, corresponds to domains that
possess rationally parameterizable domain boundaries.  It has been
generalized to the broader class of RE curves \cite{Bizzarri2016},
which share the rational envelope property with MPH curves.

Our results are also related to earlier work on methods for circle
skinning, see \cite{Kunkli2010} and the references therein. Given a
sequence of circles in a defined admissible position, these methods
find two $G^1$ continuous curves (skins) that touch each circle at a
point. The method of Kunkli and Hoffmann \cite{Kunkli2010} identifies
points of contact with associated tangents, which are then
interpolated by cubic curves, resulting in two $G^1$ continuous
skins. While this construction is rather general, in some cases, the
skins can intersect each other or the input circles. The problem of
self--intersections has been addressed in \cite{Kruppa2019}.

Alternatively, the skins can be obtained with the help of envelope
curves. In two recent papers \cite{Kruppa2020, Kruppa2021}, the
authors use RE curves to obtain rational skins without
self--intersections. Skinning can be generalized to branched sets of
circles, allowing more complicated shapes to be considered.  Bastl et
al. \cite{Bastl2015} consider unions of skins of subsets of circles,
and Kruppa et al.~\cite{Kruppa2019} analyze branching in more detail,
ensuring smooth transitions.

We present methods for generating supersets of the envelopes of
one--parameter families of circles. These supersets are either $C^0$
or $C^1$ continuous curves, which may contain
self--intersections. However, this approach allows us to impose fewer
restrictions on the admissible positions of the circles, and
self-intersections can be easily handled using the line-sweep
algorithm. Furthermore, our approach generalizes process of skinning
circles defined by a curve in $\mathbb R^{2,1}$ to skinning of a
two--parameter system of circles, i.e., an evolving domain represented
by a surface in $\mathbb R^{2,1}$. From this surface, we extract the
relevant one--parameter subsets, that form the superset of the
envelope of the evolving domain.

We approximate the envelopes of evolving domains by circular arcs. Circular arcs are well--established primitives for approximating boundaries of planar domains and they approximate a smooth curve segment with approximation order 3 \cite{Meek1995}. Biarc construction provides a $G^1$ smooth interpolation of a given curve and there are various techniques for biarc approximation of smooth curves \cite{Meek1996,Meek1995} or interpolation of discrete data \cite{Bertolazzi202066,Meek1992}.
In \cite{Sir2006}, the approximation properties of the different techniques were examined. 

Moreover, circular arcs provide many computational advantages \cite{Aichholzer2011}. The construction of Voronoi diagrams with circular arcs as sites \cite{Aichholzer2010VD} and of medial axis for planar shapes with arc spline boundaries \cite{Aichholzer2009,Aichholzer2010} is efficient and robust. For an arc spline domain, i.e., a domain represented by its boundary which consists of circular arcs, it was shown that the arc fibration kernel has again a piecewise circular boundary and a geometric algorithm that computes the arc fibration kernel was proposed \cite{weiss21}.
Self-intersections in offsets were detected and trimmed using biarc approximations of planar curves \cite{Kim2012}. In \cite{Han2019}, an effective algorithm for the Minkowski sum computation of two planar objects with an arc spline boundary was proposed, using interior disc culling. Last but not least, the line sweep algorithm \cite{Bentley1979} can be adapted for computing intersections of circular arcs with the same time complexity as for intersections of line segments.

\subsection{Outline of the method}
We propose an algorithm for arc spline approximation of envelopes of evolving domains. We assume the evolving domains are represented by evolving MATs. The method is based on the following key observations. 
\begin{enumerate}
    \item The envelopes of \emph{evolving worms} can be approximated by arc splines. A planar domain whose MAT is a curve in $\mathbb R^{2,1}$ is called a \textit{worm}. More generally, a worm is any planar domain which is the cyclographic image of a single curve segment in $\mathbb R^{2,1}$. When these domains move and evolve in time, they are called evolving worms. We characterize the envelope of the evolving worms and present methods for approximating them by arc splines.
    \item At every instance of a time parameter, an evolving free--form domain is a union of worms. The envelope of the evolving free--form domain is a subset of the union of the envelopes of evolving worms.
    \item The outer envelope can by extracted by trimming the arc splines. This can be done efficiently by the line sweep algorithm \cite{Bentley1979}.
\end{enumerate}
In this paper, we focus mostly on the first observation. In Section \ref{sec_preliminaries}, we recall Minkowski space $\mathbb R^{2,1}$, Minkowski circles and arcs and we review the medial axis transform. In Section \ref{sec_evolvingworms} we introduce worms and evolving worms and characterize their envelopes. The computation of the envelope relies on interpolation of certain curves in the Minkowski space. We present and compare two pairs of methods for interpolating the curves in the Minkowski space in Section \ref{sec_interpolationmethods}. In all four cases, the worms defined by the interpolants have arc spline boundaries. Section \ref{sec_envComp} describes the computation of the envelope of an evolving worm. In Section \ref{sec_freeform} we apply computation of envelopes of evolving worms to envelopes of evolving free--form domains and present several examples.

\section{Preliminaries}\label{sec_preliminaries}

In this section, we recall Minkowski space $\mathbb R^{2,1}$ and the medial axis transform. We interpret the branches of the medial axis transform as curves in the Minkowski space.

\subsection{Minkowski space $\mathbb R^{2,1}$}\label{Sec2}

The Minkowski space $\mathbb R^{2,1}$ is a real 3-dimensional vector space equipped with the Minkowski inner product, defined for two vectors
\begin{equation}
    \mathbf{v}_1 = (x_1,y_1,r_1)^T, \mathbf{v}_2 = (x_2,y_2,r_2)^T \in \mathbb R^{2,1} \ ,
\end{equation}
by
\begin{equation}
    \langle \mathbf{v}_1,\mathbf{v}_2\rangle_M =\mathbf{v}_1^TG\mathbf{v}_2 = x_1x_2 + y_1y_2 - r_1r_2 \ ,
\end{equation}
where $G = \text{diag}(1,1,-1)$. For a vector $\mathbf{v} \in \mathbb R^{2,1}$, we denote
\begin{equation}
  \|\mathbf v\|_M = \langle \mathbf{v},\mathbf{v}\rangle_M \ .
\end{equation}

The Minkowski inner product defines three types of vectors in $\mathbb R^{2,1}$:
\begin{definition}
 A vector $\mathbf{v} \in \mathbb R^{2,1}$ is called
\begin{itemize}
 \item \textit{space--like}, if $\|\mathbf v\|_M > 0$,
  \item \textit{light--like}, if $\|\mathbf v\|_M = 0$, or
 \item \textit{time--like}, if $\|\mathbf v\|_M < 0$.
\end{itemize}
\end{definition}

Similarly, we distinguish three types of planes in the Minkowski space:
\begin{definition}
 A plane in $\mathbb R^{2,1}$ is called \textit{space--like}, \textit{light--like} or \textit{time--like}, if the restriction of the quadratic form defined by $G$ on the plane is positive definite, indefinite degenerate or indefinite non--degenerate, respectively.
\end{definition}


There is a one--to--one correspondence between planar oriented circles and points in Minkowski space $\mathbb R^{2,1}$. An oriented circle $\mathbf c$ centered at the point $(x,y)^T \in \mathbb R^2$ with a radius $r \in \mathbb R$ is represented by the pair $\mathbf{c} = ((x,y)^T,r)$.

The mapping
\begin{align*}
  \zeta: \mathbb R^{2,1} &\rightarrow \mathbb R^2 \\
  C = (x,y,r)^T &\mapsto c = ((x,y)^T, r) \ ,
\end{align*}
that assigns an oriented circle to a point in Minkowski space $\mathbb R^{2,1}$ is called the \textit{cyclographic mapping}. 
The image of point $(x,y,0)^T \in \mathbb R^{2,1}$ is the point $(x,y)^T \in \mathbb R^2$, which is considered to be a circle with zero radius.

A curve $C(u)=(x(u),y(u),r(u))^T,\ u \in I \subset \mathbb R$, in the Minkowski space defines a one--parameter family of oriented circles in $\mathbb R^2$. 
The cyclographic image of each point $C(u_0),\ u_0 \in I,$ is an oriented circle $c(u_0)$,
\begin{equation*}
    \zeta(C(u_0)) =  c(u_0) = ((x(u_0),y(u_0))^T,r(u_0)) \ .
\end{equation*}
The cyclographic image of the curve $C(u),\ u \in I$, is then defined as
\begin{equation*}
    \zeta(C(u)) = \bigcup_{u\in I} c(u) \ .
\end{equation*}

Depending on the tangent vectors of the curve $C(u)$, the one--parameter family of circles $c(u)$ may have a real envelope. If every tangent vector of $C(u)$ is time--like, the one--parameter family $c(u)$ does not possess a real envelope. The envelope has a single component if all of tangent vectors of $C(u)$ are light--like. In the case when $C(u)$ has only space--like and isolated light--like tangents, the one--parameter family possesses a real envelope with two branches $e_C^+(u)$ and $e_C^-(u),\ u \in I,$ that can be parameterized by 
\begin{equation}\label{eq_envPar}
    e_C^{\pm}(u) = \left(\begin{array}{c}
         x\\y \end{array}\right) - r \frac{r^\prime\left(\begin{array}{c}x^\prime\\y^\prime \end{array}\right) \pm \sqrt{x^{\prime2} +y^{\prime2} - r^{\prime2}}\left(\begin{array}{c}y^\prime\\-x^\prime \end{array}\right)}
         {x^{\prime2}+y^{\prime2}} \ ,
\end{equation}
where $x,y$ and $r$ are functions of the parameter $u \in I$, see \cite{Choi97}.

A point $e_C^{\pm}(u_0)$ for some $u_0 \in I$ is a singular point of the envelope curve $e_C^{\pm}(u)$, if
\begin{equation*}
 \left.\frac{\partial e_C^{\pm}(u)}{\partial u} \right|_{u=u_0} = \left(\begin{array}{c}
         0\\0 \end{array}\right) \ .
\end{equation*}

\begin{remark}
  If the radius function $r(u),\ u \in I$, is constant, the singular points are the points $e_C^{\pm}(u_0)$ that satisfy
\begin{equation*}
 |\kappa(u_0)| = \frac{1} {|r|} \ ,
\end{equation*}
where $\kappa(u),\ u \in I,$ is the signed curvature of the spine curve $(x(u),y(u))^T,\ u \in I$, and $r = r(u),\ u \in I,$ is the constant radius. In particular, if
\begin{equation*}
 \text{sign}(\kappa(u_0)) \cdot \text{sign}(r) = 1,
\end{equation*}
the singular point is the point $e_C^+(u_0)$, otherwise it is the point $e_C^-(u_0)$.
\end{remark}

\medskip
The cyclographic image of a curve segment $C(u)\subset \mathbb R^{2,1}$ parameterized over a closed interval $[a,b] \subset \mathbb R$ defines a domain $\Omega \subset \mathbb R^2$,
\begin{equation*}
    \Omega = \zeta(C(u)) = \bigcup_{u\in [a,b]} c(u) \ .
\end{equation*}
The envelope curves $e_C^{\pm}(u),\ u \in [a,b],$ (if they exist) and the cyclographic images of the endpoints $C(a)$ and $C(b)$ form a \emph{pre--boundary} of the domain $\Omega$. Its boundary $\delta \Omega$ is then a subset of the pre--boundary, obtained by trimming the self--intersections, see Figure \ref{fig_mCurve}.

\begin{figure}[t]
\centering
\begin{minipage}{0.33\textwidth}
\centering
\includegraphics[width=0.95\linewidth]{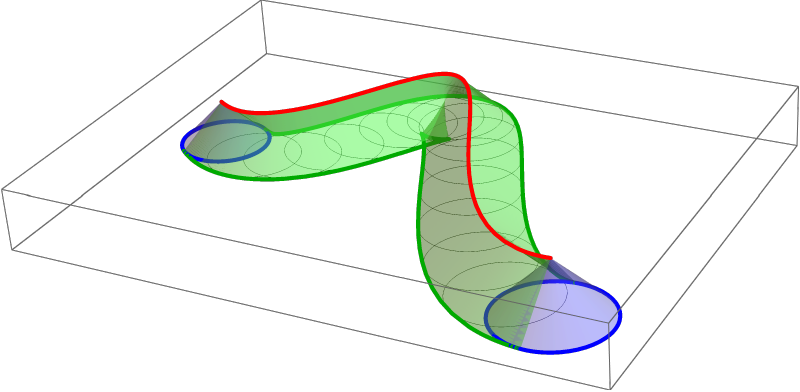}
\end{minipage}%
\begin{minipage}{0.33\textwidth}
  \centering
  \includegraphics[width=0.95\linewidth]{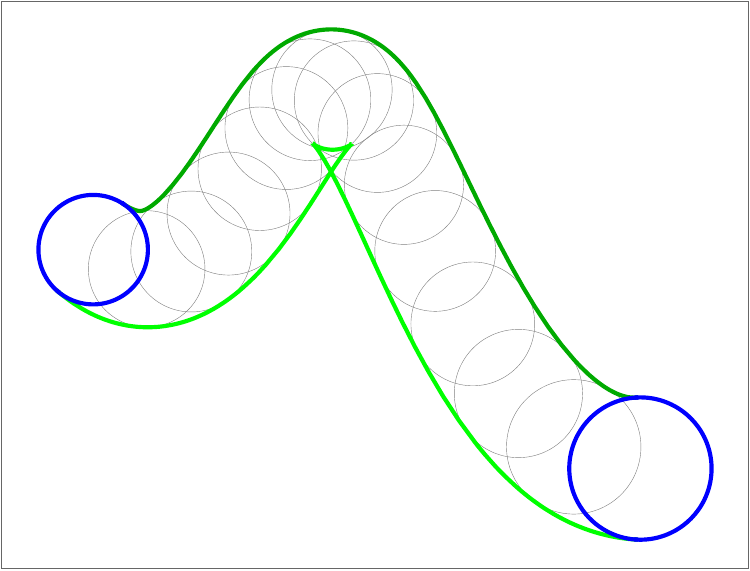}
\end{minipage}%
\begin{minipage}{0.33\textwidth}
  \centering
  \includegraphics[width=0.95\linewidth]{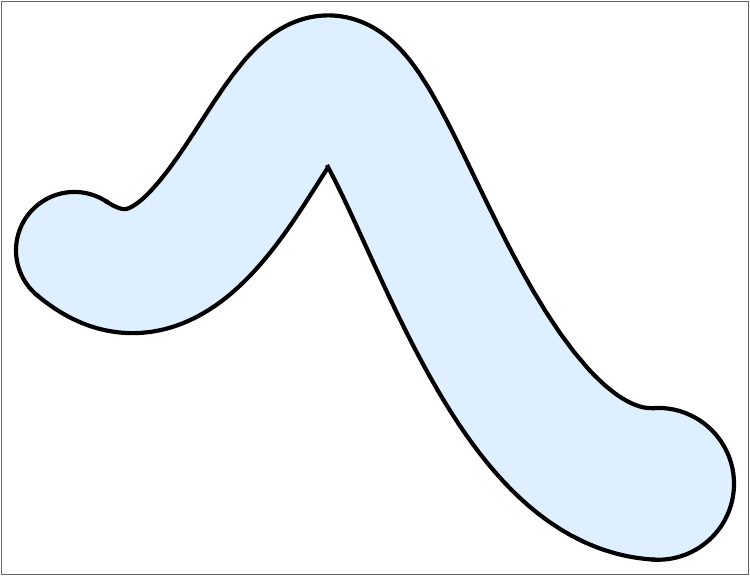}
\end{minipage}
  \caption{The cyclographic image of a curve segment in $\mathbb R^{2,1}$ is a planar domain (left). Its pre--boundary consists of the envelope curves (green and dark green) and cyclographic images of the endpoints of the curve segment (blue). The boundary of the domain is a subset of the pre--boundary (right).}
  \label{fig_mCurve}
\end{figure}

\subsection{Minkowski circles and arcs}

Two planar oriented circles $c_1$, $c_2$ are in \emph{oriented contact} if they have a common point and share a tangent vector at this point. The corresponding points $\mathbf{C}_1,\mathbf{C}_2 \in \mathbb R^{2,1}$ have zero Minkowski distance, i.e., $\| \mathbf{C}_1-\mathbf{C}_2\|_M = 0$.

A one parameter family of planar oriented circles that are all in oriented contact with two given oriented circles, defines a curve in the Minkowski space called a~\emph{Minkowski circle}. We extend this definition to more general situations. Furthermore, we review the rational parameterization for arcs of the Minkowski circles.

\medskip
We consider the complex projective extension of Minkowski space $\mathbb R^{2,1}$. For a point $\mathbf C = (x,y,r)^T \in \mathbb R^{2,1}$, we introduce its homogeneous coordinates $\tilde{\mathbf{C}}=(\bar w:\bar x:\bar y:\bar r)$.

The \textit{absolute conic} is defined as the set of points at infinity, i.e., points that satisfy the equations
\begin{align*}
    \tilde x^2+\tilde y^2-\tilde r^2 & =0 \ , \\
    \tilde w &= 0 \ .
\end{align*}
A Minkowski circle $\mathcal C$ is a non--degenerate conic that shares two points (with multiplicity) with the absolute conic. The plane in which the Minkowski circle $\mathcal C$ lies determines its affine type. If the plane is
\begin{itemize}
    \item space--like, $\mathcal C$ is an ellipse (a),
    \item time--like, $\mathcal C$ is a hyperbola with only space--like (b), or only time--like tangents (c),
    \item light--like, $\mathcal C$ is a parabola (d).
\end{itemize}

The envelope of the one parameter family of oriented circles given by the Minkowski circle $\mathcal C$ is composed of two oriented circles in the case (a) and (b), of an oriented circle and an oriented line, in the case (d), or, in the case (c), the family does not have a real envelope. Conversely, two oriented circles define a Minkowski circle of type (a) or (b). An oriented circle and an oriented line defines a Minkowski circle of type (d).

\begin{figure}[h]
\captionsetup[subfigure]{labelformat=empty}
\centering
\begin{subfigure}[b]{0.25\textwidth}
\centering
\includegraphics[width=0.95\linewidth]{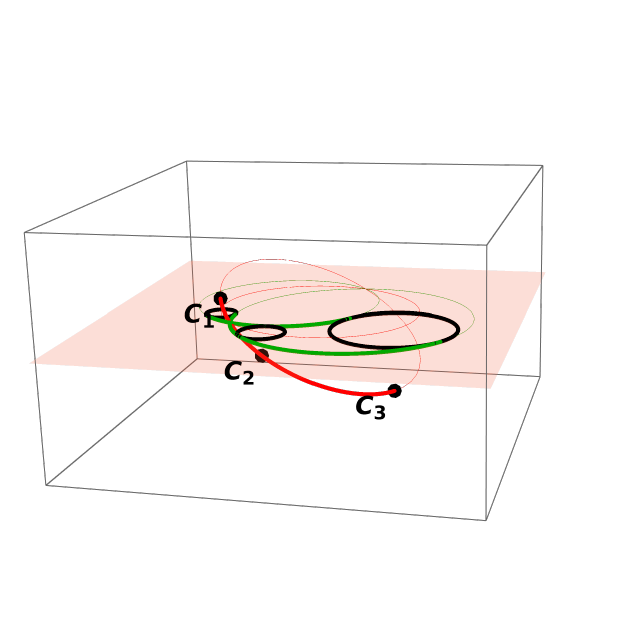}
\end{subfigure}%
\begin{subfigure}[b]{0.25\textwidth}
  \centering
  \includegraphics[width=0.95\linewidth]{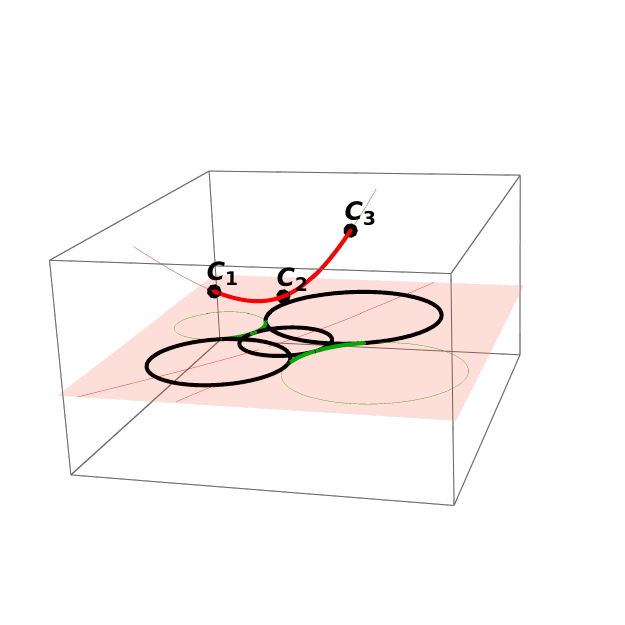}
\end{subfigure}%
\begin{subfigure}[b]{0.25\textwidth}
\centering
\includegraphics[width=0.95\linewidth]{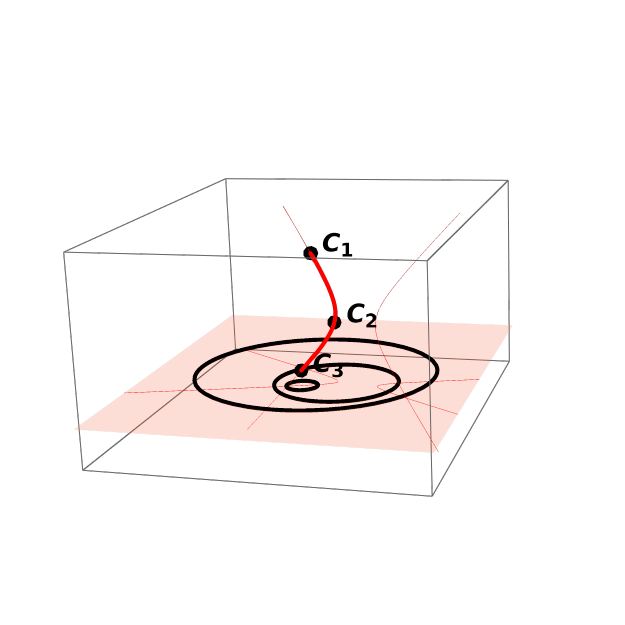}
\end{subfigure}%
\begin{subfigure}[b]{0.25\textwidth}
  \centering
  \includegraphics[width=0.95\linewidth]{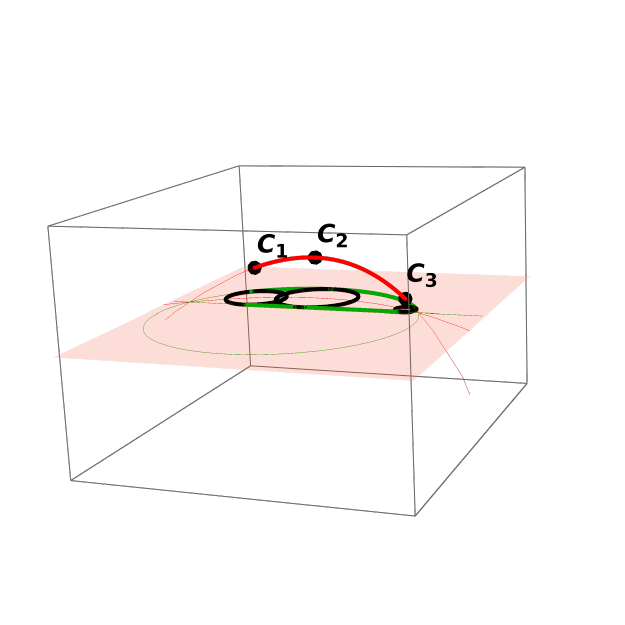}
\end{subfigure}
\begin{subfigure}[b]{0.25\textwidth}
\centering
\includegraphics[width=0.95\linewidth]{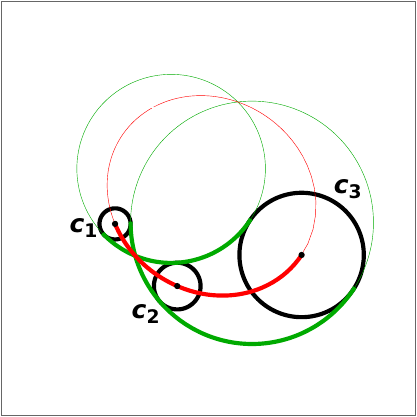}
\caption{case (a)}
\end{subfigure}%
\begin{subfigure}[b]{0.25\textwidth}
  \centering
  \includegraphics[width=0.95\linewidth]{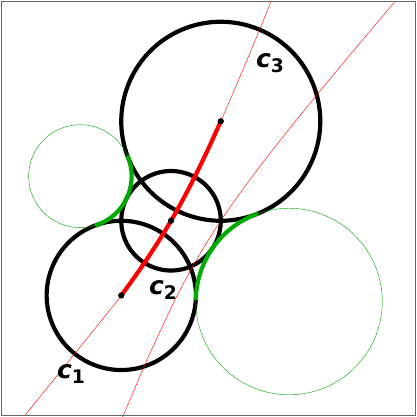}
  \caption{case (b)}
\end{subfigure}%
\begin{subfigure}[b]{0.25\textwidth}
\centering
\includegraphics[width=0.95\linewidth]{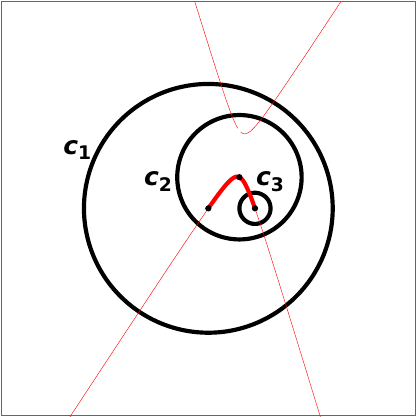}
\caption{case (c)}
\end{subfigure}%
\begin{subfigure}[b]{0.25\textwidth}
  \centering
  \includegraphics[width=0.95\linewidth]{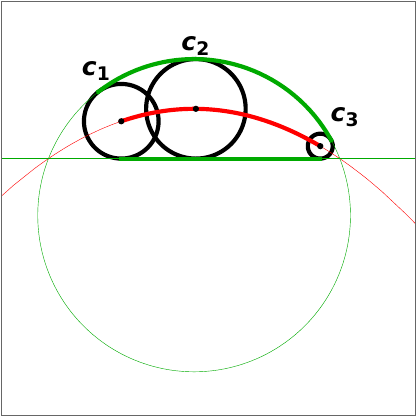}
  \caption{case (d)}
\end{subfigure}
\par\bigskip
  \caption{Affine types of Minkowski arcs and circles (red) and envelopes of the corresponding one parameter families of oriented circles (green).}
  \label{fig_minkArcs}
\end{figure}

\medskip
A Minkowski circle is uniquely determined by three distinct points  $\mathbf C_1,\mathbf C_2,\mathbf C_3\in \mathbb R^{2,1}$, which define an arc on the Minkowski circle. We propose a quadratic parameterization of the Minkowski arc in the case when the difference vectors between the three points are not light--like.

\begin{proposition}\label{prop_MinkArc}
    Let $\mathbf C_1,\mathbf C_2,\mathbf C_3 \in \mathbb R^{2,1}$ be three points in the Minkowski space whose difference vectors are not light--like. Then
\begin{equation}
    A(u) = \frac{w_1\mathbf C_1(u-1)(u-\frac 1 2)+w_2\mathbf C_2u(u-1)+w_3\mathbf C_3u(u-\frac 1 2)}{w_1(u-1)(u-\frac 1 2)+w_2u(u-1)+w_3u(u-\frac 1 2)}, \ u \in [0,1] \ ,
\end{equation}
    where
\begin{equation}
    w_1 = 2\| \mathbf C_2-\mathbf C_3\|_M \ , \quad
    w_2 = -\| \mathbf C_1-\mathbf C_3\|_M \ , \quad
    w_3 = 2\| \mathbf C_1-\mathbf C_2\|_M \ ,
\end{equation}
is the Minkowski arc interpolating $\mathbf C_1,\mathbf C_2$ and $\mathbf C_3$.
\end{proposition}

\begin{proof}
 We embed the points $\mathbf C_1,\mathbf C_2,\mathbf C_3$ into the projective extension of Minkowski space $\mathbb R^{2,1}$ by
 \begin{equation*}
   \mathbf C_i = (x_i,y_i,r_i)^T \mapsto \tilde{\mathbf C}_i = (1:x_i:y_i:r_i) \ ,\quad i = 1,2,3 \ ,
 \end{equation*}
 and interpolate the points $\tilde{\mathbf C}_i, i = 1,2,3,$ by a Lagrange polynomial of degree 2 with the help of weights $a,b \in \mathbb R$, $a,b\neq 0$,
 \begin{align*}
  \tilde A(u) &= \tilde{\mathbf C}_1(u-1)(u-\frac 1 2)+a\tilde{\mathbf C}_2u(u-1)+b\tilde{\mathbf C}_3u(u-\frac 1 2) \\
  &= (\tilde a_w(u):\tilde a_x(u):\tilde a_y(u):\tilde a_r(u))\ , \quad u \in [0,1] \ . 
 \end{align*}
 
Let $\mathbb P \mathbb R$ denote the projectively extended real line. The curve $\tilde A(u),\ u \in [0,1]$, is a Minkowski arc if there are two points $\tilde A(u_1), \tilde A(u_2)$ on the curve $\tilde A(u),\ u \in \mathbb P \mathbb R$, that lie on the absolute conic, i.e., the parameters $u_1,u_2$ satisfy the equations
 \begin{align*}
     \tilde a_x^2(u)+\tilde a_y^2(u)-\tilde a_r^2(u) & =0 \ , \\
    \tilde a_w(u) &= 0\ ,
 \end{align*}   
  where the first one is of degree 4 and the second one of degree 2 in $u$. Hence there exist some $q_0,q_1,q_2 \in \mathbb R$ such that
  \begin{equation*}
     \tilde a_x^2(u) + \tilde a_y^2(u) - \tilde a_r^2(u) - (q_0+q_1u+q_2u^2)\tilde a_w(u) = 0
  \end{equation*}
  Collecting the coefficients at $u^i,\ i = 0,1,\dots,4$, produces a system of 5 equations, from which we eliminate the parameters $q_0,q_1,q_2$. From the resulting equations we compute the weights $a$ and $b$, where the only possible combination where both of them are non--zero is
  \begin{equation*}
    a = -\frac{\| \mathbf C_1-\mathbf C_3\|_M}{2\|\mathbf C_2-\mathbf C_3\|_M} = \frac{w_2}{w_1}\ , \quad b = \frac{\| \mathbf C_1-\mathbf C_2\|_M}{\|\mathbf C_2-\mathbf C_3\|_M} = \frac{w_3}{w_1} \ ,
  \end{equation*}
 which concludes the proof. \qedhere
\end{proof}

In \cite{Pottmann1998} a quadratic rational B\'ezier parameterization of the Minkowski arc with three control points $P_1,P_2,P_3$ was suggested for the cases (a) and (b), when the Minkowski arcs have only space--like tangents. We propose an equivalent formulation.

\begin{proposition}\label{prop_MinkBezierArc}
    Let $\mathbf P_1,\mathbf P_2,\mathbf P_3 \in \mathbb R^{2,1}$ be three points in the Minkowski space such that $\|P_1-P_3\|_M > 0$ and $\|P_1-P_2\|_M > 0$ and $\|P_2-P_1\|_M  = \|P_3-P_2\|_M$. Then
\begin{equation*}
 A(u) = \frac{P_1(1-u)^2+wP_2 2 u(1-u)+P_3u^2}{(1-u)^2+w 2 u(1-u)+u^2} \ ,
\end{equation*}
where
\begin{equation*}
 w = \frac{\sqrt{\|P_1-P_3\|_M}}{2\sqrt{\|P_1-P_2\|_M}} \ ,
\end{equation*}
is the Minkowski arc satisfying
\begin{equation}\label{eq_proof21}
 A(0)= P_1\ , \quad A(1) = P_3 \ ,
\end{equation}
and
\begin{equation}\label{eq_proof22}
 \frac{A'(0)}{\|A'(0)\|} = \frac{P_2-P_1}{\|P_2-P_1\|_M} \ , \quad
 \frac{A'(1)}{\|A'(1)\|} = \frac{P_3-P_2}{\|P_3-P_2\|_M} \ .
\end{equation}
\end{proposition}

\begin{proof}
 Following the same procedure as in the previous proof, we obtain the weight $w$. The weight is real and non--zero only if $\|P_1-P_3\|_M > 0$ and $\|P_1-P_2\|_M > 0$. The equations \eqref{eq_proof21} and \eqref{eq_proof22} follow from the properties of quadratic rational B\'ezier curves.
\end{proof}

\medskip
Next we consider biarcs in the Minkowski space. Given two points $P_1,P_2\in \mathbb R^{2,1}$ and two unit space--like vectors $\mathbf t_1, \mathbf t_2$,
\begin{equation*}
 \|\mathbf t_1\|_M = \|\mathbf t_2\|_M = 1 \ ,
\end{equation*}
we can construct a Minkowski biarc interpolating the data $(P_1,\mathbf t_1)$ and $(P_2,\mathbf t_2)$.
The biarc consists of two Minkowski arcs $A_j^1(u)$ and $A_j^2(u),\ u\in [0,1],$ that satisfies
\begin{align*}
  A^1(0) &= P_1 \ , \quad A^2(1) = P_2\ , \nonumber \\
  \frac{A^{1'}(0)}{\|A^{1'}(0)\|_M} &= \mathbf t_1\ , \quad \frac{A^{2'}(1)}{\|A^{2'}(1)\|_M} = \mathbf t_2\ ,
\end{align*}
and are joined continuously with continuous tangents,
\begin{equation*}
   A^1(1) = A^2(0) \quad \text{and} \quad
   \frac{A^{1'}(1)}{\|A^{1'}(1)\|_M} = \frac{A^{2'}(0)}{\|A^{2'}(0)\|_M} \ .
\end{equation*}

We compute the Minkowski biarc following the construction from \cite{Pottmann1998}. The biarc is given by 5 B\'ezier control points
 \begin{equation}\label{eq_bezierCP}
   \left(P_1,B_1,J,B_2,P_2\right)\ ,
 \end{equation}
 where
 \begin{equation*}
  B_1 = P_1 + \lambda_1 \mathbf t_1\ , \quad B_2 = P_2 - \lambda_2 \mathbf t_2\ ,\quad \lambda_1,\lambda_2 \in \mathbb R,\ \lambda_1+\lambda_2 \neq 0 \ ,
 \end{equation*}
  such that
  \begin{equation}\label{eq_MinkBiArc}
   \langle B_2 - B_1, B_2 - B_1 \rangle_M = (\lambda_1 + \lambda_2)^2,
  \end{equation}
  i.e., the points $P_1,B_1,J$ and $J,B_2,P_2$ form isosceles triangles (with respect to the Minkowski inner product). The joint point $J$ of the two Minkowski arcs is given by
  \begin{equation*}
   J = \frac{\lambda_2B_1 + \lambda_1B_2}{\lambda_1+\lambda_2} \ .
  \end{equation*}
  Choosing the parameter $\lambda_1$, equation \eqref{eq_MinkBiArc} gives a unique parameter $\lambda_2$. If $\lambda_1,\lambda_2 >0$, the joint point $J$ lies between $B_1$ and $B_2$ and the Minkowski arcs lie inside of the triangles $P_1,B_1,J$ and $J,B_2,P_2$. 
 
  Furthermore, equation \eqref{eq_MinkBiArc} ensures that the difference vectors $J-B_1$ and $B_2 - J$ are space--like for any $\lambda_1 \neq 0$:
  \begin{equation*}
    J - B_1 = \frac{\lambda_1(B_2-B1)}{\lambda_1 + \lambda_2} \ ,
  \end{equation*}
  and hence 
  \begin{equation*}
    \| J - B_1 \|_M = \frac{\lambda_1^2\|B_2-B_1\|_M}{(\lambda_1 +\lambda_2)^2} = \lambda_1^2 \ ,
  \end{equation*}
  and similarly, $\| B_2 - J \|_M = \lambda_2^2.$ Therefore, we can use Proposition \ref{prop_MinkBezierArc} to parameterize the Minkowski biarc and its cyclographic image has a real envelope.
  
In this paper, we chose $\lambda_1 > 0$ such that
\begin{equation*}
    \| P_1 - J\|_M = \| P_2 - J\|_M \ ,
\end{equation*}
following the ‘equal chord’ biarc interpolation scheme for planar data \cite{Meek1995}.  

Minkowski arc splines are defined to be curves in Minkowski space $\mathbb R^{2,1}$ composed of arcs of Minkowski circles.

\medskip
Let $C(u) = (x(u),y(u),r(u))^T,\ u \in [0, 1],$ be a Minkowski arc with only space--like (or isolated light--like) tangents. We note that the envelope of the one parameter family of circles defined by $C(u)$ is formed by circular arcs or line segments. Two arcs are parts of~the cyclographic images of the endpoints of the Minkowski arc. The other two arcs can be constructed as follows: 
\begin{enumerate}
    \item Let $\boldsymbol\tau(u)$ be unit tangent vector at $C(u)$,
    \begin{equation*}
        \boldsymbol\tau(u) = \frac{C'(u)}{\sqrt{\|C'(u)\|_M}} \ .
    \end{equation*}
    \item For $u^* \in \{0,1\}$, we find the two light--like planes $\rho^{\pm}_{u^*}$ such that
    \begin{itemize}
        \item $C(u^*) \in \rho^{\pm}_{u^*}$, and
        \item $\boldsymbol\tau(u^*) \in \rho^{\pm}_{u^*}$.
    \end{itemize}
    \item For $u^* \in \{0,1\}$, we find the lines $\ell^{\pm}_{u^*}$ as the intersection of the planes $\rho^{\pm}_{u^*}$ and the plane $r = 0$,
    \begin{equation}\label{eq_tangentLineEnv}
        \ell^{\pm}_{u^*}(v) = \f q^{\pm}_{u^*} + \f t^{\pm}_{u^*}v \ , \quad v \in \mathbb R \ ,
    \end{equation}
    such that $\f q^{\pm}_{u^*} \in \mathbb R$, $\|\f t^{\pm}_{u^*}\| = 1$, and the lines $\ell^{\pm}_{u^*}(v)$ are in oriented contact with the circle $\zeta(C(u^*))$, i.e., their unit tangent vectors coincide at the common points.
    \item For $u^* \in \{0,1\}$, we compute the points
    \begin{equation}\label{eq_envPts}
        \f p^{\pm}_{u^*} = (x(u^*),y(u^*))^T + r(u^*)(\f t^{\pm}_{u^*})^{\perp} \ ,
    \end{equation}
    where $(\f t^{\pm}_{u^*})^{\perp}$ denotes the vector $(\f t^{\pm}_{u^*})$ rotated by $\frac{\pi}{2}$ counterclockwise.
    
    Those are exactly the envelope points $e^{\pm}_C(u^*)$ from equation \eqref{eq_envPar}.
    \item The rotations mapping the pairs $(\f p^{\pm}_{0}, \f t^{\pm}_{0})$ to the pairs
    $(\f p^{\pm}_{1}, \f t^{\pm}_{1})$ define the envelope arcs. 
    
\end{enumerate}

\subsection{Medial axis transform}
We recall the medial axis and the medial axis transform of an open set  $\Omega \subset \mathbb R^2$.
\begin{definition}
  Let $\Omega$ be an open set in $\mathbb R^2$ and $\delta \Omega$ its boundary. A closed disk $D \subset \mathbb R^2$ (possibly with a zero radius) is a \emph{maximal disk} inscribed in $\Omega$, if
\begin{itemize}
 \item $D \subseteq \Omega$, and if
 \item for any disk $D'$ satisfying $D \subseteq D'$ and $D' \subseteq \Omega$ it holds $D' = D$.
\end{itemize}
The boundary $\delta D$ of a maximal disk is called a \emph{maximal circle}.
\end{definition}

In $\mathbb R^2$ with the Euclidean metric, every maximal disk intersects $\delta\Omega$ in at least two points or it has a third order contact at points with stationary curvature.

\begin{definition}
  Let $\Omega$ be an open set in $\mathbb R^2$. The \emph{medial axis} of $\Omega$ is the union of the centres of all maximal disks inscribed in $\Omega$,
  \begin{equation*}
      \MA(\Omega) = \bigcup_{\substack{D\\ \text{maximal disk}}} \{(x_0,y_0)^T:\delta D\ \text{has the equation}\ (x-x_0)^2 + (y-y_0)^2 = r_0^2\} \ .
  \end{equation*}
\end{definition}

The medial axis transform $\MAT(\Omega)$ of the set $\Omega$ also contains the information about radii of the maximal disks.

\begin{definition}
  Let $\Omega$ be an open set in $\mathbb R^2$. The \emph{medial axis transform} of $\Omega$ is defined as
  \begin{equation*}
      \MAT(\Omega) = \bigcup_{\substack{D\\ \text{maximal disk}}} \{(x_0,y_0,r_0)^T: \delta D \ \text{has the equation}\ (x-x_0)^2 + (y-y_0)^2 = r_0^2\} \ .
  \end{equation*}
\end{definition}

We interpret the branches of $\MAT(\Omega)$ as curves $B_1(u),\dots,B_n(u)$, $n \in \mathbb N$, in Minkowski space $\mathbb R^{2,1}$ parameterized over a closed interval $I \subset \mathbb R$. The domain $\Omega$ is the union of the cyclographic images of the curves $B_i(u)$,
\begin{equation*}
    \Omega = \bigcup_{i=1}^{n} \zeta(B_i(u)) \ ,
\end{equation*}
and its boundary is a subset of the boundaries of the domains $\zeta(B_i(u))$,
\begin{equation*}
    \delta \Omega \subseteq \bigcup_{i=1}^{n} \delta\zeta(B_i(u)) \ .
\end{equation*}
Note that if the branches $B_1(u),\dots,B_n(u)$ of the MAT are Minkowski arcs, then $\Omega$ is an arc spline domain, i.e., a domain whose boundary is composed of circular arcs and line segments.

\section{Envelopes of evolving worms}\label{sec_evolvingworms}
In this section, we introduce worms and evolving worms. We characterize the supersets of the envelopes of the evolving worms as the boundaries of cyclographic images of certain curves in the Minkowski space.

\subsection{Evolving worms}

First, we define worms with the help of the cyclographic mapping:
\begin{definition}\label{def_worm}
A planar domain $\Omega \subset \mathbb R^2$ is called a \emph{worm} if it is the cyclographic image of a curve segment $B(u): I \rightarrow \mathbb R^{2,1}$, where $I \subset \mathbb R$ is a closed interval,
\begin{equation*}
    \Omega = \zeta(B(u))\ .
 \end{equation*}
\end{definition}

The boundary $\delta \Omega$ is the boundary of the cyclographic image of the curve $B(u)$. If $I = [a,b]$, then it is composed of (parts of) the curves $e^{\pm}_{B}(u),\ u \in [a,b]$, and of arcs of the cyclographic images of points $B(a)$ and $B(b)$, see Figure \ref{fig_worm}.

\begin{remark}
  By Definition \ref{def_worm}, any planar domain $\Omega$ can be considered as a worm. We can define $\Omega$ to be a \emph{strong worm} if it holds
  \begin{equation*}
    \MAT(\Omega) = B(u) \ , \quad u \in I \ ,
  \end{equation*}
  where $B(u): I \rightarrow \mathbb R^{2,1}$ and $I \subset \mathbb R$ is a closed interval.
\end{remark}

\begin{figure}[t]
\centering
\begin{minipage}{0.5\textwidth}
\centering
\includegraphics[width=0.95\linewidth]{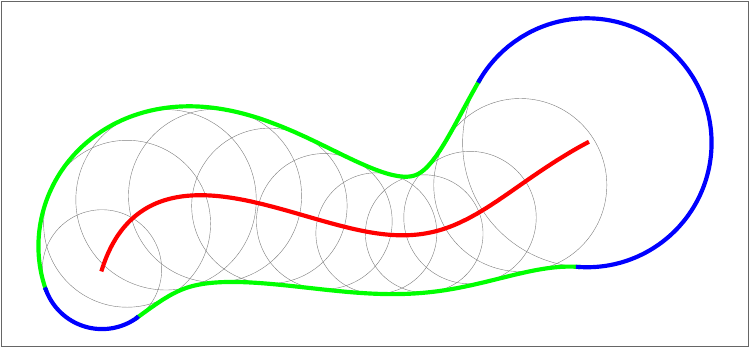}
\end{minipage}%
\begin{minipage}{0.5\textwidth}
  \centering
  \includegraphics[width=0.95\linewidth]{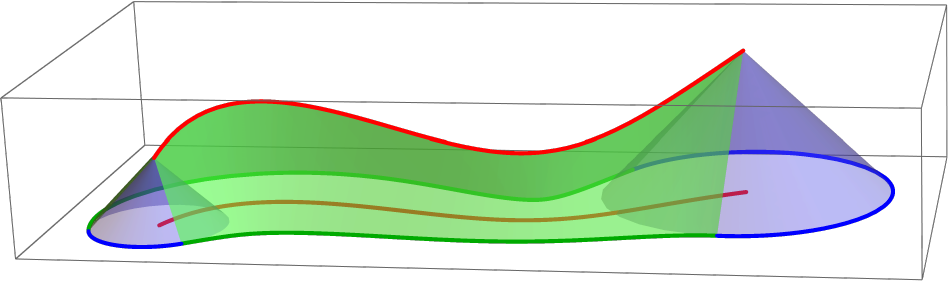}
\end{minipage}
  \caption{A worm is a cyclographic image of a curve segment in Minkowski space $\mathbb R^{2,1}$.}
  \label{fig_worm}
\end{figure}

When a worm $\Omega$ evolves over time $t$ and for each value of the time parameter its MAT is a single curve in $\mathbb R^{2,1}$, we call it an evolving worm.

\begin{definition}
  An \emph{evolving worm} is a one--parameter family of planar domains $\Omega(t),\ t \in I_t,$ such that for every $t_0 \in I_t$ the domain $\Omega(t_0)$ is a worm, i.e.,
    \begin{equation*}
        \Omega(t_0) = \zeta(B_{t_0}(u)) = \zeta(B(u,t_0))\ ,
    \end{equation*}
  where $I_u,I_t \subset \mathbb R$ are closed intervals.
\end{definition}

The union of the preimages $B_{t}(u)$ of the evolving worms $\Omega(t)$ over the time parameters $t \in I_t$ is a surface in Minkowski space $\mathbb R^{2,1}$,
\begin{equation*}
  \bigcup_{t \in I_t} B_t(u) = \{B(u,t), \ u \in I_u, \ t \in I_t \}\ .
\end{equation*}

\begin{remark}
    We can define a \emph{strong evolving worm} to be a one--parameter family of planar domains $\Omega(t),\ t \in I_t,$ such that for every $t_0 \in I_t$ the domain $\Omega(t_0)$ is a strong worm, i.e.,
 \begin{equation*}
   \MAT(\Omega(t_0)) = \{B_{t_0}(u)\ , \  u \in I_u\} = \{B(u,t_0) \ , \  u \in I_u\} \ .
  \end{equation*}
    The union of the medial axis transforms of the strong evolving worms is again a surface in Minkowski space $\mathbb R^{2,1}$,
\begin{equation*}
  \bigcup_{t \in I_t} \MAT(\Omega(t)) = \{B(u,t), \ u \in I_u, \ t \in I_t \}\ .
\end{equation*} 

\end{remark}

The envelope of the evolving worm $\Omega(t),\ t \in [0,1],$ is the envelope of a two--parameter family of oriented planar circles. If  
\begin{equation*}
    B(u,t)=(x(u,t),y(u,t),r(u,t))^T,\ u,t \in [0,1] \ ,
\end{equation*}
the circles are centred at $(x(u,t),y(u,t))^T$ and their radii are given by the function $r(u,t)$.

\subsection{Envelope characterization}
We characterize the envelope of an evolving worm. Let $\Omega(t),\ t \in [0,1]$, be an evolving worm represented by a smooth surface $B(u,t)$ in Minkowski space $\mathbb R^{2,1}$.
\begin{equation*}
 B(u,t) = (x(u,t),y(u,t),r(u,t))^T, \quad u,t \in [0,1] \ .
\end{equation*}

We identify \textit{singular} curves on the surface $B(u,t)$. The cyclographic images of the singular curves contribute to the envelope of the evolving worm.

\begin{definition}\label{def_singPt}
 Let $B(u,t),\ u,t \in [0,1],$ be a surface in Minkowski space $\mathbb R^{2,1}$. A point $B(u_0,t_0),$ $u_0,t_0~\in~[0,1],$ is a \emph{singular point} of $B(u,t),$ if
 the vectors
 \begin{equation*}
 \mathbf v_u(u_0,t_0) = \left.\frac{\partial B(u,t)}{\partial u} \right|_{u=u_0,t=t_0} \quad \text{and} \quad
 \mathbf v_t(u_0,t_0) = \left.\frac{\partial B(u,t)}{\partial t} \right|_{u=u_0,t=t_0} 
\end{equation*}
 are either linearly dependent or if they span a light--like plane.
\end{definition}

Let $u_0,t_0 \in [0,1]$. The restriction of the quadratic form defined by $G = \text{diag}(1,1,-1)$ on the tangent plane to $B(u,t)$ at the point $B(u_0,t_0)$ is defined by the $2\times 2$ matrix
\begin{equation*}
    M(u_0,t_0) = \left(\mathbf v_u,\mathbf v_t\right)^T G \left(\mathbf v_u,\mathbf v_t \right) \ ,
\end{equation*}
where $\mathbf v_u = \mathbf v_u(u_0,t_0)$ and $\mathbf v_t = \mathbf v_t(u_0,t_0)$.

The tangent plane is light--like if the quadratic form $M(u_0,t_0)$ is indefinite degenerate, that is, 
\begin{equation}\label{eq_det}
    \text{det}(M(u_0,t_0)) = 0 \ .
\end{equation}
Equation \eqref{eq_det} identifies the points of the surface $B(u,t)$ with light--like tangent planes as well as the second type of the singular points from Definition \ref{def_singPt}, i.e., the points where the tangent plane is not defined.

If $B(u,t)$ is piecewise rational surface, Equation \eqref{eq_det} defines an algebraic curve in the parameter domain. Points $B(u_0,t_0)$ for $u_0,t_0 \in [0,1],$ that satisfy Equation \eqref{eq_det} then form the \emph{singular curves} on the surface $B(u,t)$.

The envelope of the evolving worm $\Omega(t)$ then consists of boundaries of the cyclographic images of the singular curves and the \textit{boundary curves} of the surface $B(u,t)$. This fact was noted in \cite{Peternell2009}. Here we present a formal proof.

\begin{figure}[t]
\centering
\begin{minipage}{0.55\textwidth}
\centering
\includegraphics[width=0.95\linewidth]{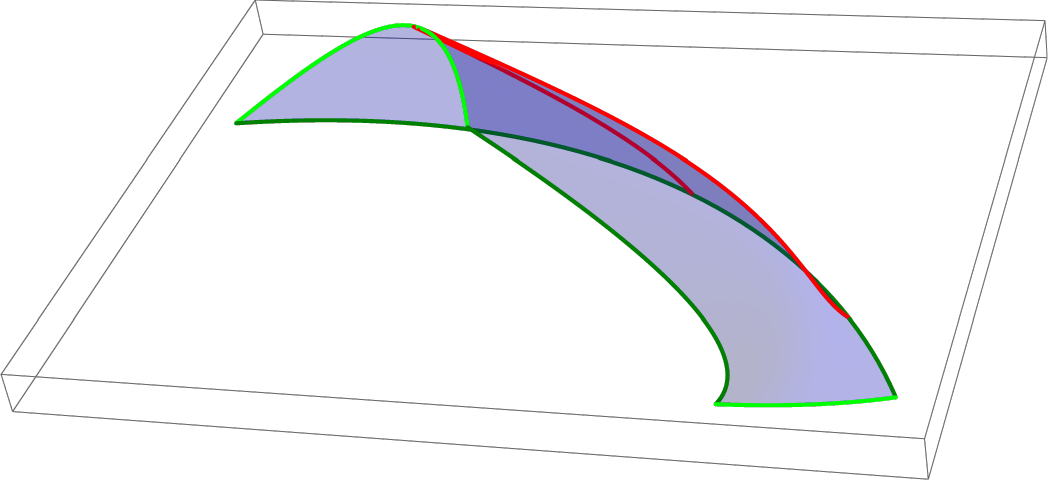}
\end{minipage}%
\begin{minipage}{0.45\textwidth}
  \centering
  \includegraphics[width=0.95\linewidth]{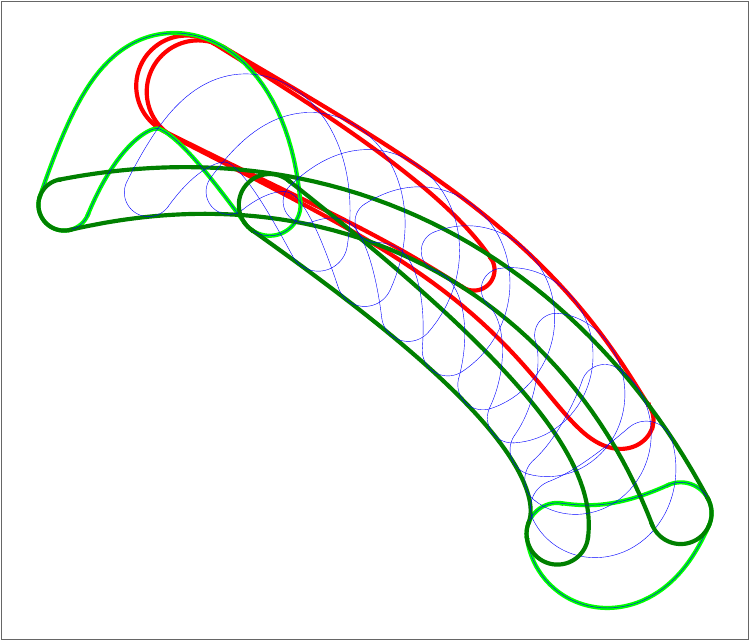}
\end{minipage}
  \caption{The singular (red) and the boundary (dark and light green) curves on a surface in the Minkowski space (left) and their cyclographic images (right).}
\label{fig_singCurves}
\end{figure}

\begin{theorem}\label{thm_singularCurves}
 The envelope of an evolving worm $\Omega(t),\ t \in [0,1],$  given by the surface
 \begin{equation*}
   B(u,t):[0,1]^2 \rightarrow \mathbb R^{2,1}\ ,
 \end{equation*}
  is a subset of the cyclographic images of the boundary curves
 \begin{align*}
  & B(u,0),\ B(u,1),\quad u \in [0,1] \ ,\\
  & B(0,t),\ B(1,t),\quad t \in [0,1] \ ,
 \end{align*}
 and of the cyclographic images of the singular curves.
\end{theorem}

\begin{proof}
 We denote $B(u,t) = (x(u,t),y(u,t),r(u,t)),\ u,t \in [0,1]$. The surface defines a two parameter family of planar circles $p$, which can be parameterized with the help of three parameters $u,t \in [0,1]$ and $\varphi \in [0,2\pi)$,
 \begin{align*}
    p(u,t,\varphi)=\left(\begin{array}{c}
         x(u,t)\\y(u,t) \end{array}\right)+r(u,t)\left(\begin{array}{c}
         \cos{\varphi}\\ \sin{\varphi} \end{array}\right) \ .
 \end{align*}
The boundary curves correspond to the families of circles
\begin{align*}
  & p(u,0,\varphi),\ p(u,1,\varphi),\quad u \in [0,1],\ \varphi \in [0,2\pi)\ ,\\
  & p(0,t,\varphi),\ p(1,t,\varphi),\quad t \in [0,1],\ \varphi \in [0,2\pi)\ ,
 \end{align*}
and the cyclographic images of the boundary curves are envelopes of these families of circles, that clearly correspond to the envelope of $\Omega(t)$.

Let $u,t\in (0,1)$. Any point of the envelope of the two parameter family of circles $p(u,t,\varphi)$ satisfies the envelope condition (see e.g. \cite{Peternell2009}):
\begin{equation*}
  \text{rank}\left(
  \frac{\partial p}{\partial u} ,
  \frac{\partial p}{\partial t} ,
  \frac{\partial p}{\partial \varphi}
  \right) < 2 \ .
\end{equation*}
The Jacobian matrix is rank deficient iff
\begin{equation*}
  \text{det}\left(
  \frac{\partial p}{\partial u} ,
    \frac{\partial p}{\partial \varphi}
  \right) =0 \quad \text{and } \quad  \text{det}\left(
  \frac{\partial p}{\partial t} ,
    \frac{\partial p}{\partial \varphi}
  \right) =0 \ ,
\end{equation*}
i.e., if the following equations hold
\begin{align*}
 r x_u \cos{\varphi}
 + r r_u \cos^2{\varphi}
 + r y_u \sin{\varphi}
 + r r_u \sin^2{\varphi} &= 0 \ ,\\
  r x_t \cos{\varphi}
 + r r_t \cos^2{\varphi}
 + r y_t \sin{\varphi}
 + r r_t \sin^2{\varphi} &= 0 \ ,
\end{align*}
where $r=r(u,t)$ and
\begin{align*}
 x_u &= \frac{\partial x(u,t)}{\partial u},\quad y_u = \frac{\partial y(u,t)}{\partial u},\quad r_u = \frac{\partial r(u,t)}{\partial u}\ , \\
 x_t &= \frac{\partial x(u,t)}{\partial t},\quad y_t = \frac{\partial y(u,t)}{\partial t},\quad r_t = \frac{\partial r(u,t)}{\partial t} \ .
\end{align*}

Eliminating the parameter $\varphi$ from the equations above yields the equation
\begin{equation}
 -r \left(
 (x_u y_t - x_t y_u)^2 - (x_u r_t - x_t r_u)^2 - (y_u r_t - y_t r_u)^2
 \right)
\end{equation}
which is exactly the equation \eqref{eq_det} multiplied by $-r$. Hence we see that the points that satisfy the envelope condition are the points of the cyclographic images of the singular points and the circles of the two parameter family which have zero radii.
\end{proof}

\section{Minkowski arc interpolation of curves in $\mathbb R^{2,1}$}\label{sec_interpolationmethods}
The envelope of an evolving worm consists of boundaries of cyclographic images of certain curves in $\mathbb R^{2,1}$. We present and compare two pairs of methods for curve interpolation in the Minkowski space by Minkowski arcs either \emph{directly}, or \emph{indirectly}, by interpolating the boundaries of the cyclographic images of the original curves. In both cases, we interpolate the curve or the boundary of its cyclographic image either by (Minkowski) arcs or (Minkowski) biarcs. The resulting envelope of the evolving worm is then an arc spline.

We consider curves with space-like tangents only. Light-like tangents may be present at endpoints. The four schemes described below need to be modified to deal with these points. This is described in the appendix of the paper.  

\subsection{First pair: direct methods}
Let $\Omega \subset \mathbb R^2$ be a worm and let $C(v),\ v \in [0,1]$, be a curve in Minkowski space $\mathbb R^{2,1}$ such that
\begin{equation*}
  \Omega = \zeta(C(v)) \ .
\end{equation*}
The \emph{direct} methods interpolate the curve $C(v)$ by Minkowski (bi)arcs. The boundaries of the cyclographic images of interpolants are composed of circular arcs and approximate the boundary $\delta \Omega$. The MAT of the approximated worm is an Minkowski arc spline.

We assume that the curve $C(v)$ has only space--like or isolated light--like tangents, i.e., the worm $\Omega$ is the cyclographic image of the whole curve $C(v)$.

We choose $N \in \mathbb N$ and sample points on the curve $C(v),\ v \in [0,1]$, at $N+1$ parameter values
\begin{equation*}
  0 = v_0 < v_1 < \dots < v_N = 1 \ ,
\end{equation*}
and denote the samples
\begin{equation*}
  C_j = C(v_j),\ j = 1,\dots,N \ .
\end{equation*}
We interpolate segments of the curve $C(v)$ by Minkowski arcs (direct arc interpolation, \textbf{DAI}), or by Minkowski biarcs (direct biarc interpolation, \textbf{DBI}).
Then we compute the boundaries of the cyclographic images of the Minkowski (bi)arcs exactly. These curves are composed of circular arcs and approximate the boundary of the worm $\Omega$.

\subsection*{Direct arc interpolation}
We interpolate segments of the curve $C(v),\ v \in [0,1],$ in Minkowski space $\mathbb R^{2,1}$ by Minkowski arcs and approximate the boundary of the worm $\Omega$ by the boundaries of the cyclographic images of the Minkowski arcs.

\begin{flushleft}
Method \textbf{DAI}[$C(v),N$]\\
\textbf{Input:} \textit{Curve $C(v),v\in [0,1],$ in $\mathbb R^{2,1}$ with only space--like or isolated light--like tangents, $N\in \mathbb N$ even.}\\
\textbf{Output:} \textit{Superset of the boundary of the domain $\tilde \Omega$ that approximates the worm $\Omega$, $\MAT(\Omega) = \{C(v), v \in [0,1]\}$.}
\begin{enumerate}
 \item \textit{Sample the curve $C(v),\ v\in [0,1]$, at  $N+1$ parameter values}
 \begin{equation*}
  0 = v_0 < v_1 < \dots < v_N = 1 \ .
  \end{equation*}
 \item \textit{Compute points $C_j = C(v_j),\ j = 1,\dots,N$.}
 \item \textit{For each $k = 0,1,\dots,\frac{N}{2}-1$, interpolate the three consecutive points
 \begin{equation*}
   C_{2k},\ C_{2k+1},\ C_{2k+2} \ ,
  \end{equation*}
by a Minkowski arc $\tilde C_k(v),\ v \in [0,1]$, according to Proposition \ref{prop_MinkArc}.}
\item \textit{For each $k = 0,1,\dots,\frac{N}{2}-1$, compute the cyclographic image $\zeta(\tilde C_k(v)),\ v \in [0,1]$.}
\item \textit{The boundary of the domain $\tilde \Omega$ is a subset of the boundaries of the cyclographic images,}
\begin{equation*}
  \delta \tilde \Omega \subseteq \bigcup_{k=0}^{\frac{N}{2}-1} \left\{ \delta \zeta(\tilde C_k(v)),\ v \in [0,1] \right\}\ .
\end{equation*}
\end{enumerate}
\end{flushleft}

The approximated worm $\tilde \Omega$ is then a cyclographic image of the Minkowski arc spline
\begin{equation*}
 \tilde C(v) = \tilde C_k\left(\frac{N}{2} v - \frac{2k}{N}\right)\ ,  \quad v \in \left[\frac{2k}{N},\frac{2k+2}{N}\right]\ ,
\end{equation*}
where $k = 0,1,\dots,\frac{N}{2}-1$.
The Minkowski arc spline $\tilde C(v),\ v \in [0,1]$, is $C^0$ continuous at the points $C_{2k}$ for $k = 1,\dots,\frac{N}{2}-1$. However, the endpoints of the envelope curves of two consecutive segments $\tilde C_k(v)$ and $\tilde C_{k+1}(v)$ do not generally coincide. Therefore we need to include also the cyclographic images of the endpoints $C_k(0)$ and $C_k(1)$ of the Minkowski arc spline segments $\tilde C_k(v),\ v \in [0,1],$ for $k = 1,\dots,\frac{N-1}{2}$, to approximate the whole boundary of the worm $\Omega$.

\subsection*{Direct biarc interpolation}
We interpolate segments of the curve $C(v),\ v \in [0,1],$ in Minkowski space $\mathbb R^{2,1}$ by Minkowski biarcs and approximate the boundary of the worm $\Omega$ by the boundaries of the cyclographic images of the Minkowski biarcs.

\begin{flushleft}
Method \textbf{DBI}[$C(v),N$]\\
\textbf{Input:} \textit{Curve $C(v),v\in [0,1],$ in $\mathbb R^{2,1}$ with only space--like or isolated light--like tangents, $N\in \mathbb N$.}\\
\textbf{Output:} \textit{Superset of the boundary of the domain $\tilde \Omega$ that approximates the worm $\Omega$, $\MAT(\Omega) = \{C(v), v \in [0,1]\}$.}
\begin{enumerate}
\item \textit{Sample the curve $C(v),\ v\in [0,1]$, at  $N+1$ parameter values}
 \begin{equation*}
  0 = v_0 < v_1 < \dots < v_N = 1 \ .
  \end{equation*}
 \item \textit{Compute points $C_j = C(v_j),\ j = 1,\dots,N$.}
 \item \textit{At each point $C_j,\ j = 0,1,\dots,N$, compute the normalized tangent vector $\mathbf t_j$  of the curve $C(v)$,
  \begin{equation*}
   \mathbf t_j = \frac{C'(v_j)}{\sqrt{ \| C'(v_j)\|_M}} \ .
  \end{equation*}}
\item \textit{For each $j = 0,1,\dots,N-1$, we compute the control points of a Minkowski biarc that interpolates the consecutive pairs of data
  \begin{equation*}
    \left(C_j, \mathbf t_j \right) \quad \text{and} \quad  \left(C_{j+1}, \mathbf t_{j+1} \right)\
 \end{equation*}
and parameterize the arcs of the biarc $\tilde C_j(v),\ v\in [0,1]$, according to Proposition \ref{prop_MinkBezierArc}.}
\item \textit{For each $j = 0,1,\dots,N-1$ we compute the envelope curve}
\begin{equation*}
  e_{\tilde C_j}^{\pm}(v),\ v \in [0,1] \ .
\end{equation*}
\item \textit{The boundary of the domain $\tilde \Omega$ is a subset of the envelope curves and arcs of the cyclographic images of $C_0$ and $C_N$,}
\begin{equation*}
  \delta \tilde \Omega \subseteq \bigcup_{j=0}^{N-1} \left\{ e_{\tilde C_j}^{\pm}(v),\ v \in [0,1] \right\} \cup \zeta(C_0) \cup \zeta(C_N) \ .
\end{equation*}
\end{enumerate}
\end{flushleft}

Note that as we assume all the tangents of the curve $C(v),\ v \in [0,1]$, to be space--like, we can always normalize them.

The approximated worm $\tilde \Omega$ is then a cyclographic image of the Minkowski arc spline
\begin{equation*}
 \tilde C(v) = \tilde C_j\left(N v - \frac{j}{N}\right) \ , \quad v \in \left[\frac{j}{N},\frac{j+1}{N}\right]\ ,
\end{equation*}
where $j = 0,1,\dots,N-1$. The Minkowski arc spline is $G^1$ continuous and the envelope curves $e_{\tilde C}^{\pm}(v),\ v \in [0,1],$ are planar circular arc splines and are $C^1$ continuous at the points $C_j$ for $j = 2,\dots,N-1$.

\subsection{Second pair: indirect methods}
Let $\Omega \subset \mathbb R^2$ be a worm and let $C(v),\ v \in [0,1]$, be a curve in Minkowski space $\mathbb R^{2,1}$ such that
\begin{equation*}
  \Omega = \zeta(C(v)) \ .
\end{equation*}
The \emph{indirect} methods interpolate the boundary $\delta \Omega$, more specifically the envelope curves $e^{\pm}_C$, by planar circular (bi)arcs. The approximated worm is then a cyclographic image of a~geometric graph in $\mathbb R^{2,1}$, which has Minkowski arcs as its branches.

We assume that the curve $C(v)$ has only space--like or isolated light--like tangents. Furthermore, we assume that the envelope curves $e_C^{\pm}(v), \ v \in [0,1],$ do not have any singular points.

We choose $N \in \mathbb N$ and sample the curve $C(v),\ v\in [0,1]$, at  $N+1$ parameter values
 \begin{equation*}
  0 = v_0 < v_1 < \dots < v_N = 1 \ ,
 \end{equation*}
and we compute the corresponding envelope points
 \begin{equation*}
   e_j^{\pm}=e^{\pm}(v_j),\ j = 0,\dots,N \ .
 \end{equation*}
We interpolate the envelope points by planar circular arcs (indirect arc interpolation, \textbf{IAI}), or planar biarcs (indirect biarc interpolation, \textbf{IBI}). These curves form the boundary $\delta \tilde \Omega$ of the domain $\tilde \Omega$ which approximates the worm $\Omega$.

The arc spline $\delta \tilde \Omega$ approximates the boundary $\delta \Omega$ with approximation order 3, see \cite{Meek1995}.

\subsection*{Indirect arc interpolation}
We interpolate segments of the envelope curves $e^{\pm}_C(v)$ of the curve $C(v),v\in[0,1],$ by planar circular arcs.

\begin{flushleft}
Method \textbf{IAI}[$C(v),N$]\\
\textbf{Input:} \textit{Curve $C(v),v\in [0,1]$, in $\mathbb R^{2,1}$ with only space--like or isolated light--like tangents, such that $e^{\pm}_C(v), v\in [0,1]$, has no singular points, $N\in \mathbb N$ even.}\\
\textbf{Output:} \textit{Superset of the boundary of the domain $\tilde \Omega$ that approximates the worm $\Omega$, $\MAT(\Omega) = \{C(v), v \in [0,1]\}$.}
\begin{enumerate}
 \item \textit{Sample the curve $C(v),\ v\in [0,1]$, at  $N+1$ parameter values}
 \begin{equation*}
  0 = v_0 < v_1 < \dots < v_N = 1 \ .
 \end{equation*}
 \item \textit{For each $j = 0,1,\dots N$, compute the envelope points
 \begin{equation*}
   e_j^{\pm}=e^{\pm}(v_j),\ j = 0,\dots,N \ .
 \end{equation*}}
 \item \textit{For each $k = 0,1,\dots,\frac N 2 - 1$, interpolate the three consecutive points
  \begin{equation*}
   e^+_{2k}\ ,e^+_{2k+1}\ ,e^+_{2k+2} \ ,
  \end{equation*}
  by a planar circular arc $\tilde e^+_k(v),\ v \in [0,1]$}.
  \item \textit{The boundary of the domain $\tilde \Omega$ is a subset of the envelope curves $\tilde e^{\pm}_k(v),\ v \in [0,1],$ and arcs of the cyclographic images of $C_0 = C(0)$ and $C_N = C(1)$,}
  \begin{equation*}
  \delta \tilde \Omega \subseteq \bigcup_{k=0}^{\frac N 2 -1} \left\{\tilde e_k^{\pm}(v), v \in [0,1]\right\} \cup \zeta(C_0) \cup \zeta(C_N).
\end{equation*}
\end{enumerate}
\end{flushleft}

The arc splines
\begin{equation*}
 \tilde e^{\pm}(v) = \tilde e^{\pm}_k(v)\left(\frac{N}{2} v - \frac{2k}{N}\right) \ , \quad v \in \left[\frac{2k}{N},\frac{2k+2}{N}\right]\ ,
\end{equation*}
where $k = 0,1,\dots,\frac{N}{2}-1$, are $C^0$ continuous at the points $e^{\pm}_{2k+1},\ k = 1,\dots,\frac{N-3}{2}$.

The MAT of the domain $\tilde \Omega$ is a geometric graph whose branches are Minkowski arcs.

\subsection*{Indirect biarc interpolation}
We interpolate segments of the envelope curves $e^{\pm}(v)$ of the curve $C(v), v \in [0,1],$ by planar biarcs.

\begin{flushleft}
Method \textbf{IBI}[$C(v),N$]\\
\textbf{Input:} \textit{Curve $C(v),v\in [0,1],$ in $\mathbb R^{2,1}$ with only space--like or isolated light--like tangents, such that $e^{\pm}_C(v), v\in [0,1]$ has no singular points, $N\in \mathbb N$ even.}\\
\textbf{Output:} \textit{Superset of the boundary of the domain $\tilde \Omega$ that approximates the worm $\Omega$, $\MAT(\Omega) = \{C(v), v \in [0,1]\}$.}
\begin{enumerate}
 \item \textit{Sample the curve $C(v),\ v\in [0,1]$, at  $N+1$ parameter values}
 \begin{equation*}
  0 = v_0 < v_1 < \dots < v_N = 1 \ .
 \end{equation*}
 \item \textit{For each $j = 0,1,\dots N$, compute the envelope points}
 \begin{equation*}
   e_j^{\pm}=e^{\pm}(v_j),\ j = 0,\dots,N \ .
 \end{equation*}
 \item \textit{For each $j = 0,1,\dots,N,$ compute the unit tangent vectors $\mathbf t_j^{\pm}$ of the envelope curves $e^{\pm}_C(v)$ at the points $e_j^{\pm}$.}
  \item \textit{For every $j = 0,1,\dots,N,$ compute biarcs $\tilde e^{\pm}_j(v),\ v \in [0,1]$, interpolating the pairs}
  \begin{equation*}
    \left(e_j^{\pm},\ \mathbf t_j^{\pm}\right) \quad \text{and} \quad \left(e_{j+1}^{\pm},\ \mathbf t_{j+1}^{\pm}\right) \ .
  \end{equation*}
\item \textit{The boundary of the domain $\tilde \Omega$ is a subset of the envelope curves $\tilde e^{\pm}_j(v),\ v \in [0,1],$ and arcs of the cyclographic images of $C_0 = C(0)$ and $C_N = C(1)$,}
  \begin{equation*}
  \delta \tilde \Omega \subseteq \bigcup_{j=0}^{N -1} \left\{\tilde e_j^{\pm}(v), v \in [0,1]\right\} \cup \zeta(C_0) \cup \zeta(C_N).
\end{equation*}
\end{enumerate}
\end{flushleft}

The arc splines
\begin{equation*}
 \tilde e^{\pm}(v) = \tilde e^{\pm}_j(v)\left(\frac{N}{2} v - \frac{j}{N}\right) \ , \quad v \in \left[\frac{j}{N},\frac{j+1}{N}\right]\ ,
\end{equation*}
where $k = 0,1,\dots,N-1$, are $G^1$ continuous

In general, the MAT of the domain $\tilde \Omega$ is a geometric graph with Minkowski arcs as branches.

\begin{figure}
\captionsetup[subfigure]{labelformat=empty}
\centering
\begin{subfigure}[b]{0.5\textwidth}
\centering
\includegraphics[width=0.95\linewidth]{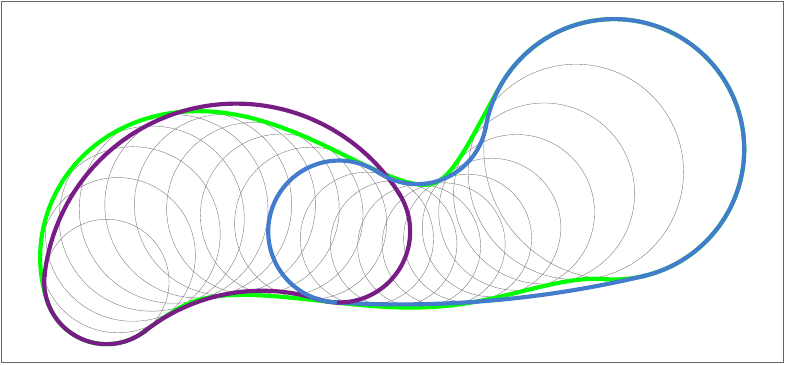}
\caption{\textbf{DAI}}
\end{subfigure}%
\begin{subfigure}[b]{0.5\textwidth}
  \centering
  \includegraphics[width=0.95\linewidth]{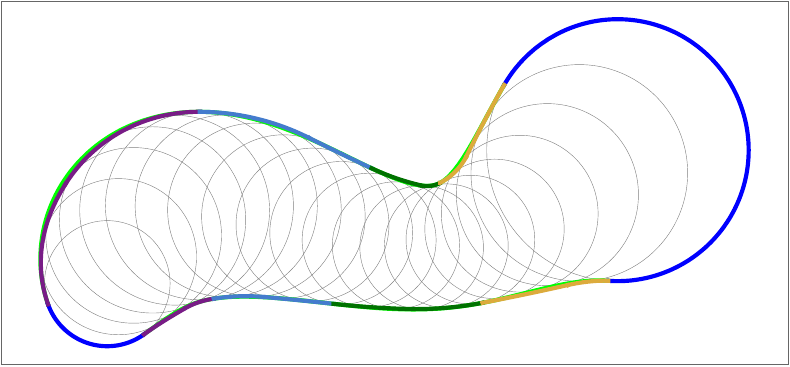}
  \caption{\textbf{DBI}}
\end{subfigure}
\par\bigskip
\begin{subfigure}[b]{0.5\textwidth}
\centering
\includegraphics[width=0.95\linewidth]{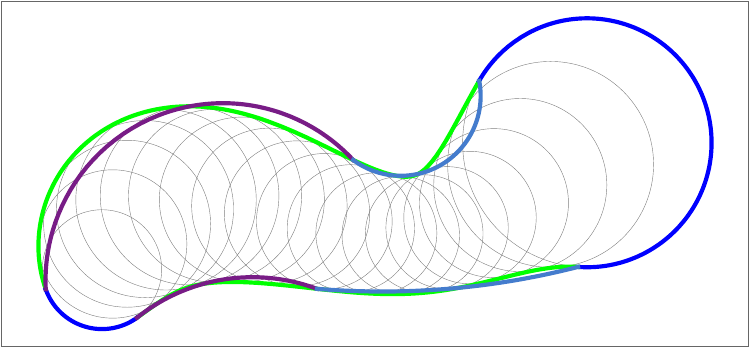}
\caption{\textbf{IAI}}
\end{subfigure}%
\begin{subfigure}[b]{0.5\textwidth}
  \centering
  \includegraphics[width=0.95\linewidth]{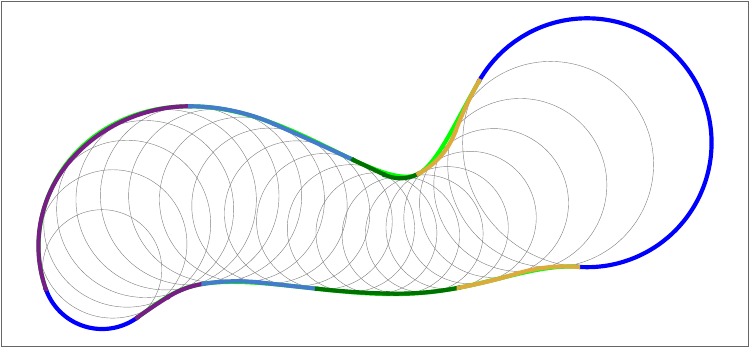}
  \caption{\textbf{IBI}}
\end{subfigure}
  \caption{Comparison of methods for approximation of curves in $\mathbb R^{2,1}$ by Minkowski arcs for $N=4$.}.
  \label{fig_envApprox}
\end{figure}

\subsection{Comparison of the methods}
We compare the four methods \textbf{DAI}, \textbf{DBI}, \textbf{IAI},
and \textbf{IBI} for the approximation of the boundary  of the worm $\Omega$ whose MAT is represented by a curve $C(v),\ v \in [0,1],$ with respect to the number of arcs, approximation power and continuity. The results are summarized in Table \ref{table_comparison}.

Let $N \in \mathbb N$ be and even number. We sample the interval $[0,1]$ uniformly and $N + 1$ parameter values and denote $h = \frac 1 N$.

\subsubsection*{Number of arcs}
The methods described above return the superset of the arcs that approximate the  boundary of worm $\Omega$. These arcs are the input for the sweep line algorithm that extracts the real envelope. We compare the cardinality of the supersets, which gives an upper bound of the number of arc in the approximated boundary.

\textbf{DAI}: We interpolate the curve $C(v)$ by $\frac N 2$ Minkowski arcs. The boundaries of the cyclographic images of the Minkowski arcs are not continuous, hence we need to include $N-1$ arcs of the circles
\begin{equation*}
  \zeta(C(ih)) \quad \text{for} \quad i = 1,2,\dots,\frac N 2 -1\ ,
\end{equation*}
and the two arcs of the circles $\zeta(C(0))$ and $\zeta(C(1))$. In total the number of arcs in the boundary $\delta \tilde \Omega$ is bounded by
\begin{equation*}
  2 \left(\frac N 2 \right) + 2\left(\frac N 2 - 1\right) + 2 = 2N \ .
\end{equation*}

\textbf{DBI}: We interpolate the curve $C(v)$ by $N$ Minkowski biarcs, whose envelopes are composed of $4N$ arcs. Adding the two arcs of the cyclographic images $\zeta(C(0))$ and $\zeta(C(1))$ of the endpoints of the curve, we obtain
$4N+2$ circular arcs in total.

\textbf{IAI}: We interpolate each of the envelope curves $e^{\pm}_C(v)$ by $\frac N 2$ arcs. Adding the two arcs of the circles $\zeta(C(0))$ and $\zeta(C(1))$ yields $N + 2$ arcs.

\textbf{IBI}: We interpolate each of the envelope curves $e^{\pm}_C(v)$ by $N$ biarcs. Together with the two arcs of the circles $\zeta(C(0))$ and $\zeta(C(1))$ we obtain $4N + 2$ arcs.

\subsubsection*{Approximation power}
The indirect methods \textbf{IAI} and \textbf{IBI} approximate the boundary of the worm $\Omega$ with approximation order 3. The direct methods \textbf{DAI} and \textbf{DBI} approximate the boundary only with approximation order 2. We consider the error with respect to the number of segments, or number of arcs (Figure \ref{fig_convergence}, right). Among the methods considered, the \textbf{IAI} method generates the lowest number of arcs for a given accuracy.

\subsubsection*{Continuity of $\delta \tilde \Omega$}
Regarding the continuity, the boundary $\delta \tilde \Omega$ is $C^0$ continuous for the methods \textbf{DAI} and \textbf{IAI} and $C^1$ continuous in the case of the methods \textbf{DBI} and \textbf{IBI}.

\begin{figure}[t]
\centering
\begin{subfigure}[b]{0.4\textwidth}
\raisebox{3.1cm}{\centering

\includegraphics[width=0.95\linewidth]{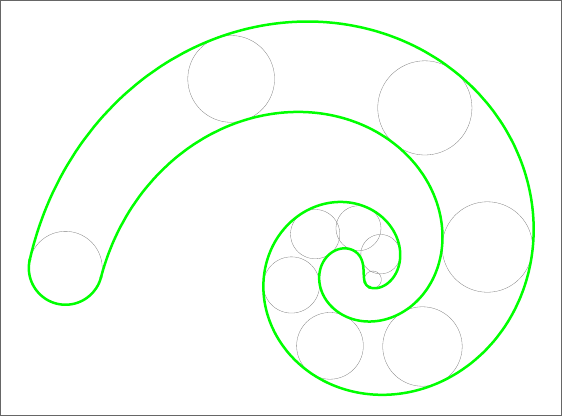}}
\end{subfigure}%
\begin{subfigure}[b]{0.6\textwidth}
  \centering
  \includegraphics[width=0.95\linewidth]{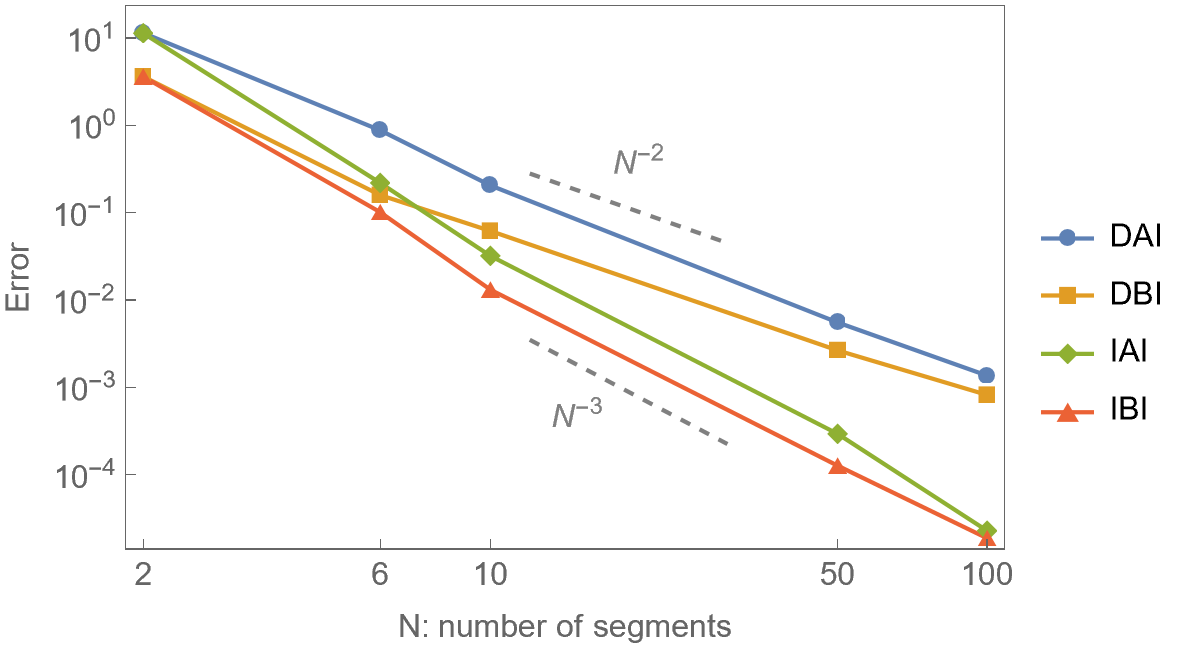}  \includegraphics[width=0.95\linewidth]{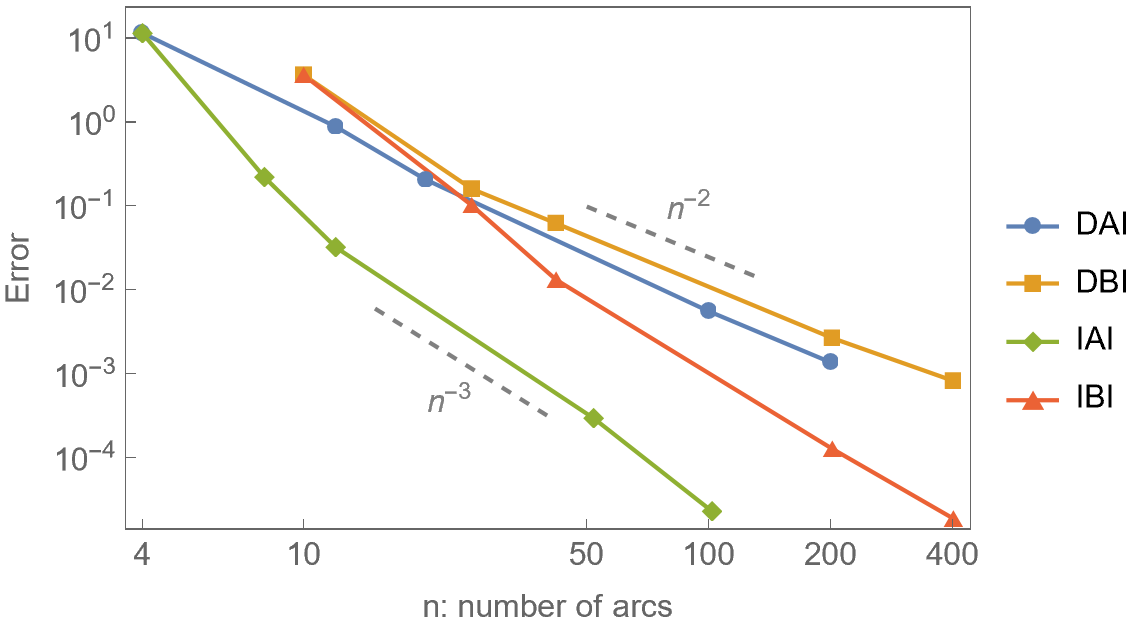}
\end{subfigure}
\caption{The indirect methods \textbf{IAI} and \textbf{IBI} approximate the boundary of a worm with approximation order 3. The direct methods \textbf{DAI} and \textbf{DBI} approximate the boundary with approximation order 2.}
\label{fig_convergence}
\end{figure}

\begin{table}[h]
\centering
\begin{tabular}{|l|cccc|}
\hline
 & \textbf{DAI} & \textbf{DBI} & \textbf{IAI} & \textbf{IBI} \\ \hline
\# of arcs & $2N$ & $4N+2$ & $N+2$ & $4N+2$ \\
continuity &  $C^0$ & $C^1$ & $C^0$ & $C^1$  \\
approximation power &  2 & 2 & 3 & 3 \\ \hline
\end{tabular}
\caption{Comparison of the methods for interpolating curves in $\mathbb R^{2,1}$.}
\label{table_comparison}
\end{table}

In summary, when deciding which method to use, we can follow the quidelines outlined here:
\begin{itemize}
  \item \textbf{IBI} offers the best option when $C^1$ smoothness is desired.
  \item \textbf{IAI} produces the lowest number of arcs for a given precission.
  \item \textbf{DAI} is the simplest method, which does not require additional effort in the case when tangent vectors of the curve $C(v) \in \mathbb R^{2,1}$ are close to light--like vectors (see Appendix for more details). 
\end{itemize}

\subsection{Adaptive sampling}
In practice, better results are obtained by sampling the curve $C(v), v \in [0,1]$, adaptively rather than uniformly.
For the sake of brevity, we focus on the \textbf{IAI} method, which was found to produce the lowest number of arcs for a given accuracy when uniform sampling was used. We present a simple adaptive procedure based on a user--defined parameter $\delta$. The adaptive procedure could be formulated similarly for the other three methods. 

\begin{flushleft}
\begin{enumerate}
 \item \textit{Sample the curve $C(v),\ v\in [v_0,v_1]$, at the parameter values $v_0, \frac{v_0+v_1}{2}, v_1$.}
 \item \textit{Compute the envelope points $e_j^{\pm}=e^{\pm}(v_j),\ j = 0,\frac 1 2,1 $}.
 \item \textit{Interpolate the three envelope points by a planar circular arcs $\tilde e^\pm(v),\ v \in [0,1]$.}
  \item \textit{Compute the maximal distances $d^\pm$ between the envelope curves $e^\pm$ and their approximations $\tilde e^\pm$}.
  \item \textit{If $d^\pm > \delta$, divide the interval in half and repeat the procedure on both intervals $[v_0,\frac{v_0+v_1}{2}]$ and $[\frac{v_0+v_1}{2},v_1]$ on the branch $e^+(v)$ or $e^-(v)$ of the envelope curves.}
\end{enumerate}
\end{flushleft}

We can obtain lower errors for the same number of arcs using the
adaptive sampling method. Figure \ref{fig_adaptive1} shows an example
of a worm approximated by the uniform \textbf{IAI} method (left, top)
and the adaptive \textbf{IAI} method (left, bottom). To achieve the
same error, the adaptive method needs less arcs. In the example, the
error is $0.05$, and the envelope obtained by sampling uniformly uses
36 circular arcs, while the adaptive approach needs only 24 arcs. A
similar behavior was observed previously in \cite{vrablikova2024} for
a different application of arc splines, where the approximation of
envelopes of envelopes of planar swept volumes can be improved by a
constant factor when adaptive sampling is used. However, this does not
affect the overall rate of convergence.

\begin{figure}[t]
\centering
\begin{subfigure}[b]{0.4\textwidth}
{\centering
\includegraphics[width=0.95\linewidth]{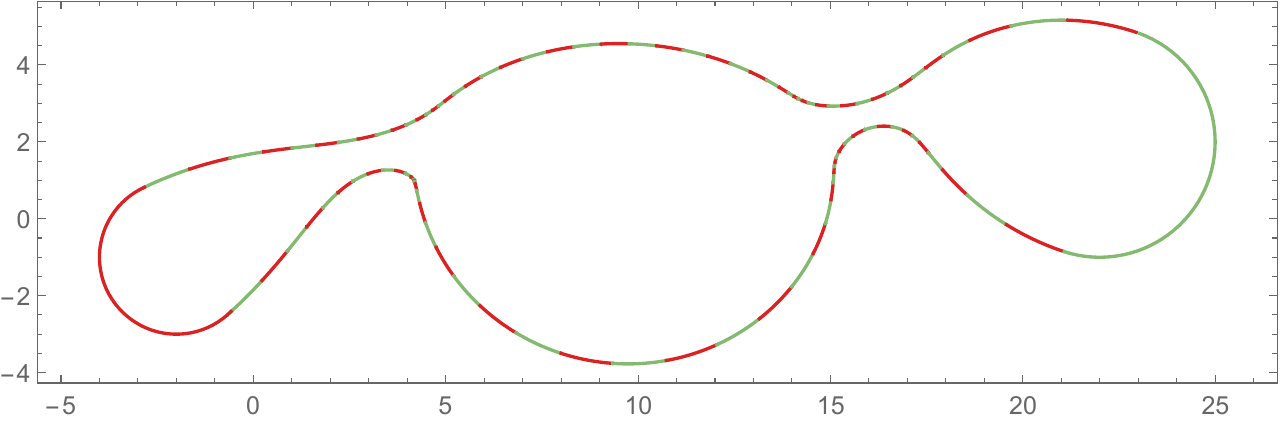}
\includegraphics[width=0.95\linewidth]{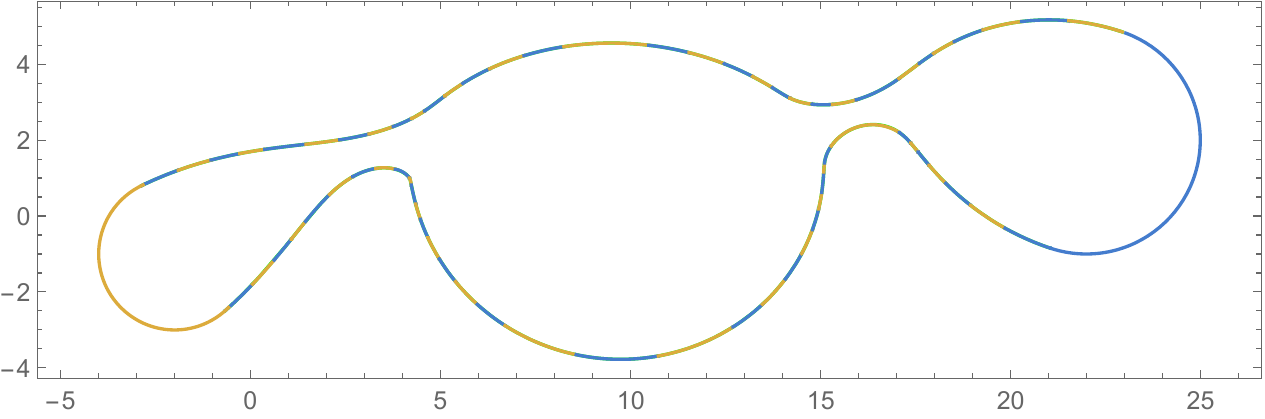}
}
\end{subfigure}%
\begin{subfigure}[b]{0.6\textwidth}
\raisebox{-5mm}{\centering 
 \includegraphics[width=0.95\linewidth]{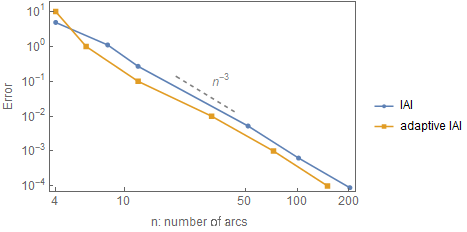}
}
\end{subfigure}
\caption{The adaptive approach for the \textbf{IAI} method results in lower errors for the same number of arcs.}
\label{fig_adaptive1}
\end{figure}

\section{Envelope computation}\label{sec_envComp}
In this section we describe the computation of the envelope of an evolving worm using one of the four method presented above. 

Let $\Omega(u,t),\ u,t \in [0,1],$ be an evolving worm defined by a piecewise smooth rational surface $B(u,t),\ u,t \in [0,1]$, in Minkowski space $\mathbb R^{2,1}$. By Theorem \ref{thm_singularCurves} we know that the envelope of $\Omega(u,t)$ is a subset of the boundaries of the cyclographic images of the boundary and the singular curves on the surface $B(u,t)$.  To approximate the boundary curves, can we simply use one of the methods \textbf{DAI}, \textbf{DBI}, \textbf{IAI} or \textbf{IBI} for the curves
\begin{align*}
  & B(u,0),\ B(u,1),\quad u \in [0,1] \ ,\\
  & B(0,t),\ B(1,t),\quad t \in [0,1] \ .
 \end{align*}
The singular curves are defined by Equation \eqref{eq_det}. The equation defines an algebraic curve in the parameter domain. Adapting a method for tracing algebraic curves from \cite{lee1999}, we trace points $(u_i,t_i)^T$ on the curve in the parameter domain and tangent vectors $\f v^i = (v_u^i,v_t^i)^T$ at those points. From this data, we recover the singular points $\f p^i=B(u_i,t_i)^T$ on the surface $B(u,t)$ and tangent vectors \begin{equation*}
    \boldsymbol\tau^i =  v^i_u \cdot \left.\left(\frac{\partial B}{\partial u}\right)^\perp\right|_{u=u_i,t = t_i} 
     +
     v^i_t \cdot  \left.\left(\frac{\partial B}{\partial t}\right)^\perp\right|_{u=u_i,t = t_i} \ .
\end{equation*}
We divide the singular curves into segments with space--like, light--like and time--like tangents and only consider those with space--like tangents and possibly light--like tangents at the endpoints. The segments with time--like tangents do not correspond to a real envelope and hence we do not need to consider them further.

\begin{figure}[t]
\centering
\begin{minipage}{0.32\textwidth}
\centering
\includegraphics[width=0.95\linewidth]{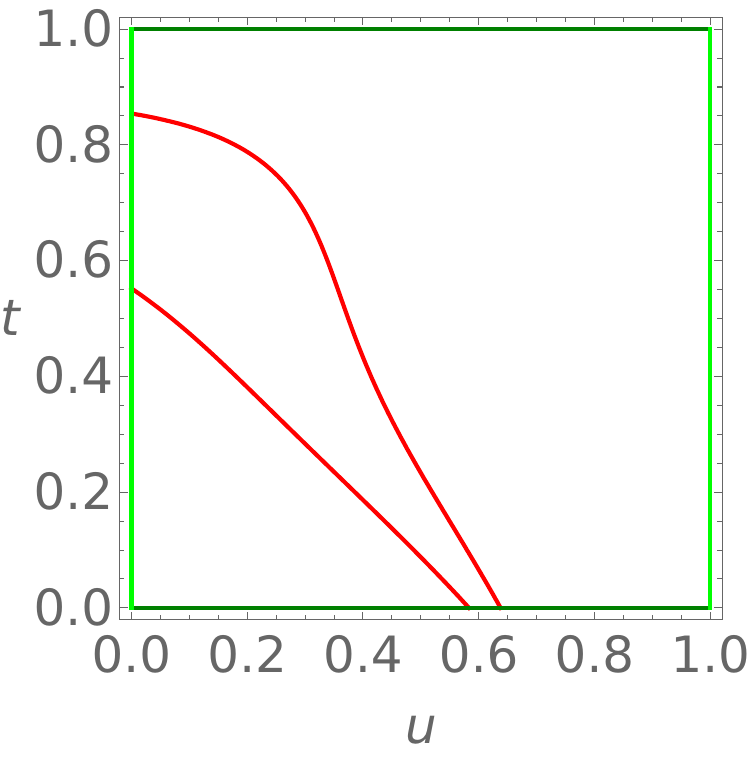}
\end{minipage}%
\begin{minipage}{0.335\textwidth}
\raisebox{5mm}{  \centering
  \includegraphics[width=0.95\linewidth]{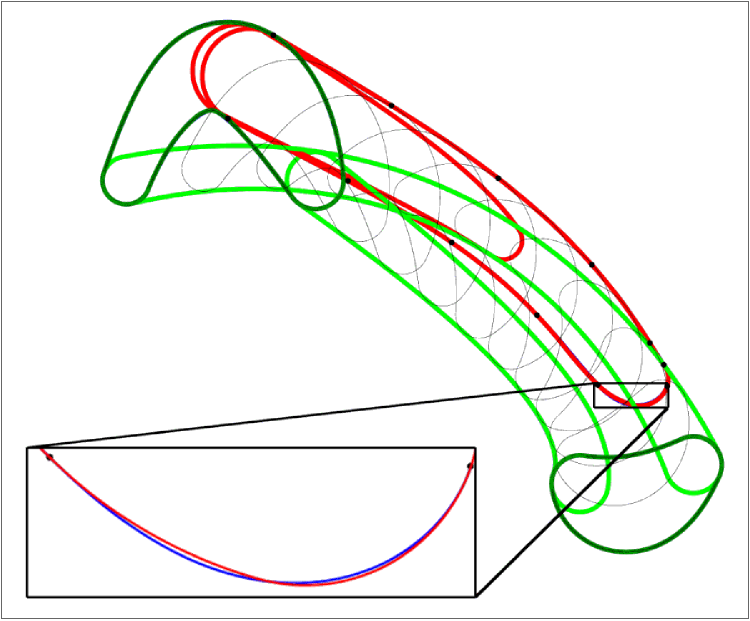}}
\end{minipage}%
\begin{minipage}{0.335\textwidth}
\raisebox{5mm}{  \centering
  \includegraphics[width=0.92\linewidth]{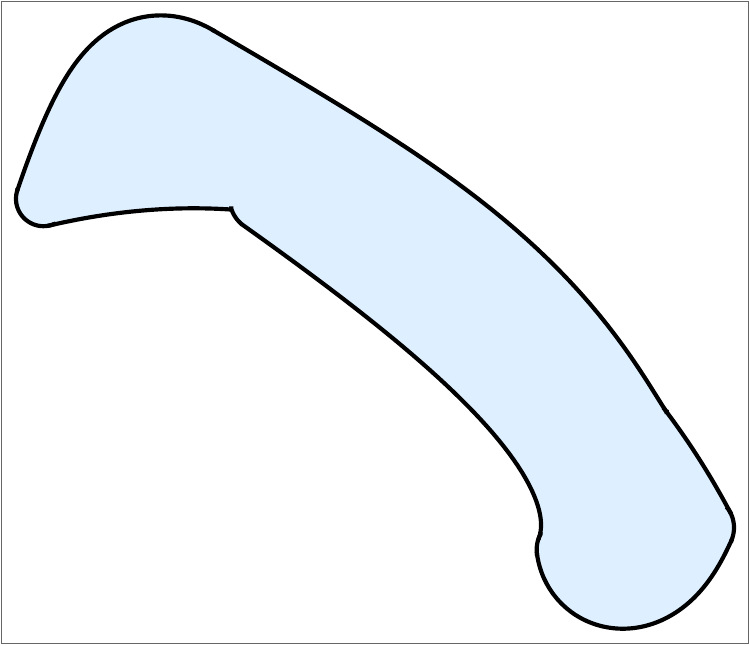}}
\end{minipage}
\caption{Curves in the parameter domain define the boundary (green) and the singular (red) curves (left). The approximation of the evolving worm using the method \textbf{IBI} and $N=4$ (centre). The boundaries of the cyclographic images of the boundary curves $B(u,0),\ B(u,1)$, the boundary curves $B(0,t),\ B(1,t)$ and the singular curves are depicted in blue, their approximations in light green, dark green and red, respectively. The envelope of the evolving domain is a subset of the boundaries of the cyclographic images (right).}
\label{fig_singCurvePlane}
\end{figure}

To approximate the singular curve using the method \textbf{DAI}, we only need the points $\f p^i$. For the other methods, we also need the tangent vectors $\boldsymbol \tau^i$. For the method \textbf{IBI}, the tangent vectors $\f t^i$ of the envelope points are also necessary, see Equations \eqref{eq_tangentLineEnv} and \eqref{eq_envPts} for the construction of the tangent vectors and the envelope points. 

Figure \ref{fig_singCurvePlane} (left) shows the singular and the boundary curves on the surface $B(u,t)$ depicted in Figure \ref{fig_singCurves} in the parameter domain. We sampled $5$ points on the curves and computed the Hermite data on the envelope curves to approximate the envelope of the envelope of the evolving worm using the method \textbf{IBI}. The outer envelope can be extracted after computing intersections of the biarcs using a line sweep algorithm, see Figure \ref{fig_singCurvePlane} (right).

\section{Envelopes of evolving free-form domains}\label{sec_freeform}

We apply the developed method for computation of envelopes of evolving worm to evolving free--form domains.

\begin{figure}[h]
  \centering
  \includegraphics[width=\linewidth]{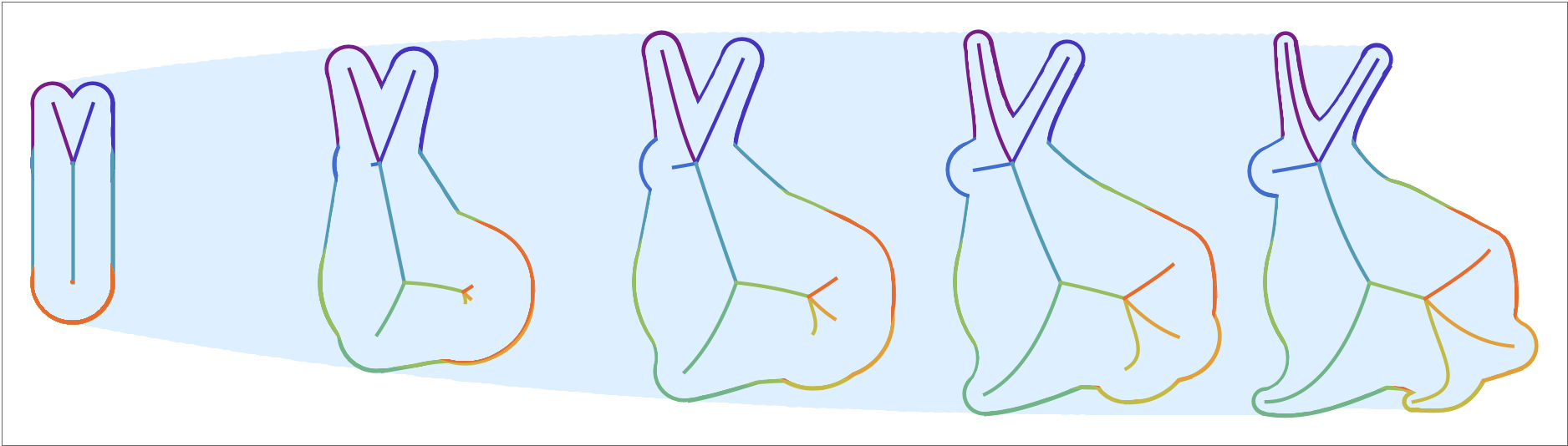}
  \caption{Evolving planar domain which is composed of 9 evolving worms.}
  \label{fig_bunny}
\end{figure}

\subsection{Free--form domains}
We interpret a free--form domain $\Omega \subset \mathbb R^2$ as a union of worms, e.g., the cyclographic images of the branches of the $\MAT(\Omega)$. 
Let $B_1(u),\dots,B_n(u)$ for some $n\in \mathbb N$ and $u\in [0,1]$ be curves in Minkowski space $\mathbb R^{2,1}$,
\begin{equation*}
    B_i(u): [0,1] \rightarrow \mathbb R^{2,1} \ .
\end{equation*}
Each curve $B_i(u)$ defines a worm $\Omega_i$,
\begin{equation*}
    \Omega_i = \zeta(B_i(u))\ , 
\end{equation*}
and the domain $\Omega$ is the union of the worms
\begin{equation*}
    \Omega = \bigcup_{i=1}^n \zeta(B_i(u)) = \bigcup_{i=1}^n \Omega_i \ .
\end{equation*}

A free--form domain $\Omega(t)$ evolving in time $t \in [0,1]$, is then represented as a union of evolving worms. Each of the evolving worms 
\begin{equation*}
    B_1(u,t),\dots,B_n(u,t)\ ,
\end{equation*}
defines a surface in the Minkowski space,
\begin{equation*}
 B_i(u,t):[0,1]^2 \rightarrow \mathbb R^{2,1} \ , \quad i = 1,\dots,n \ .
\end{equation*}

We assume that for any $i = 1,\dots, n$, there exists an interval $\bar I_i \subset [0,1]$ such that for each $t_0 \in \bar I_i$, the curve $B_i(u,t_0), \ u \in [0,1]$, has non--zero length.

\subsection{Envelope computation}
To compute the envelope of an evolving free-form domain $\Omega(t),\ t\in [0,1]$, we consider each of the evolving worms separately. The procedure consists of the following two steps:
\begin{enumerate}
\item For every $i = 1,\dots,n$, we approximate the envelope of the evolving worm $\Omega_i$ separately. That is, on each surface 
\begin{equation*}
    B_i(u,t),\quad u,t\in [0,1] \ ,
\end{equation*}
we identify the boundary and the singular curves and approximate them (and their cyclographic images) using one of the methods \textbf{DAI}, \textbf{DBI}, \textbf{IAI} or \textbf{IBI}.

\item The envelope of the evolving domain $\Omega(t)$ is a subset of the envelopes of the evolving worms $\Omega_i(t)$.
We extract the envelope of the evolving planar domain $\Omega(t)$. Using a~variant of the sweep--line algorithm for circular arcs, see \cite{Bentley1979}, we compute intersections of the envelope curves and extract the outer envelope of the evolving domain.
\end{enumerate}

\subsection{Examples}
The method for calculating the envelopes of evolving free--form domains is illustrated by the following three examples.

\begin{example}
    Consider the evolving domain from Figure \ref{fig_bunny}. At each time parameter, the domain is a union of $9$ worms. The evolving worms form surfaces 
 \begin{equation*}
     B_i(u,t):[0,1]^2 \rightarrow \mathbb R^{2,1} \ , \quad i =1,\dots,9 \ .
 \end{equation*}
Several worms evolve from a point into a curve of non--zero length.

\begin{figure}[t]
  \centering
  \begin{subfigure}{\textwidth}
    \centering
      \includegraphics[width=0.85\linewidth]{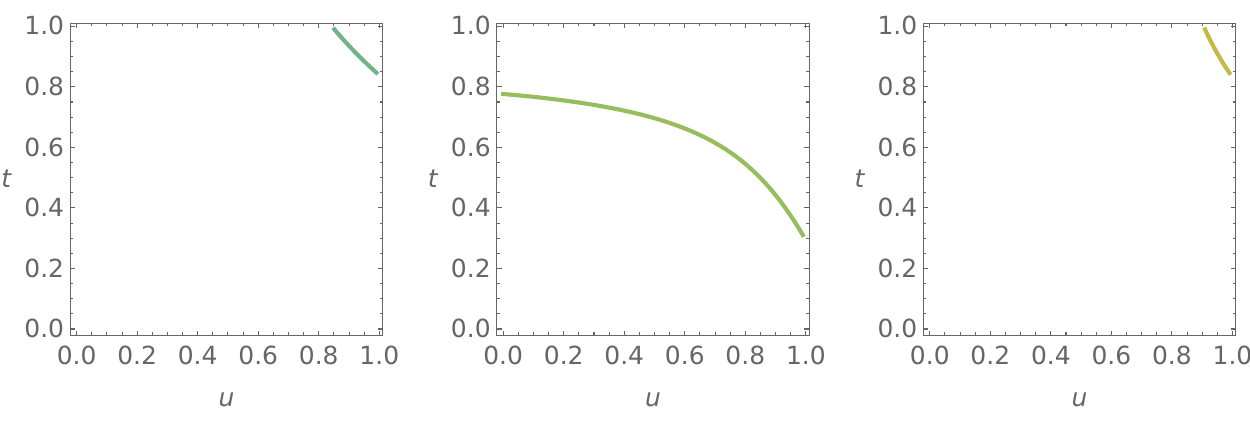}
        \caption{The singular curves in the parameter domain corresponding to three different evolving worms.}
     \label{fig_bunnySingC}
  \end{subfigure}
  
\end{figure}

\begin{figure}[t]
\ContinuedFloat
\begin{subfigure}{\textwidth}
    \centering
  \includegraphics[width=\linewidth]{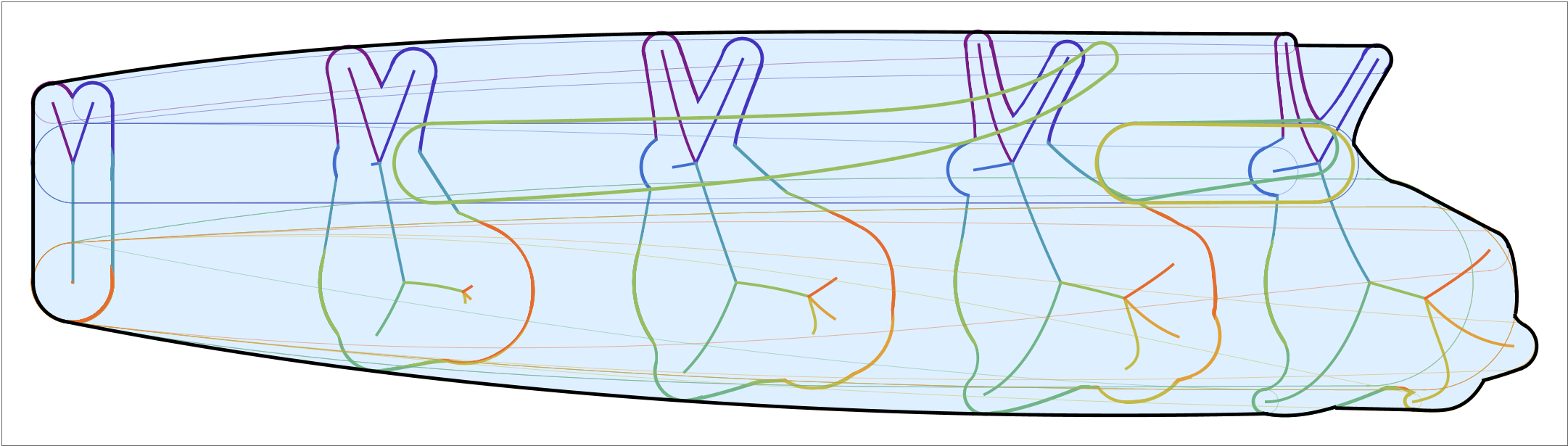}
  \caption{Superset of the envelope of the evolving domain consists of boundaries of cyclographic images of the singular (thick) and the boundary curves. The outer envelope is shown in black.}
  \label{fig_bunnyEnv}
\end{subfigure}
\caption{Example 1.}
\end{figure}

We approximate each evolving branch separately, by approximating the worms corresponding to the boundary and the singular curves of each of the surfaces
\begin{equation*}
  B_i(u,t)\ ,\quad i = 1,\dots, 9 \ .
\end{equation*}
In this example, the only evolving worms with singular curves are those shown in dark green, light green and yellow (see Figure \ref{fig_bunnySingC} and \ref{fig_bunnyEnv}). We approximate the boundaries of the cyclographic images of the singular curves and of the boundary curves using the method \textbf{IBI}. In this case, the cyclographic images of the singular curves do not contribute to the envelope of the evolving domain and are trimmed away. The envelope of the evolving domain is then extracted from the interpolations of the cyclographic images of the boundary curves.  
\end{example}

\begin{example}
Let $\Omega(t),\ t \in [0,1]$ be the evolving domain shown in Figure \ref{fig_ex2domain}. At each time instance, the domain is a union of three worms. In this case, the MAT of each worm has a non--zero length at each time parameter.

\begin{figure}[b]
\centering
\begin{subfigure}[t]{0.5\textwidth}
\centering
\includegraphics[width=0.6\linewidth]{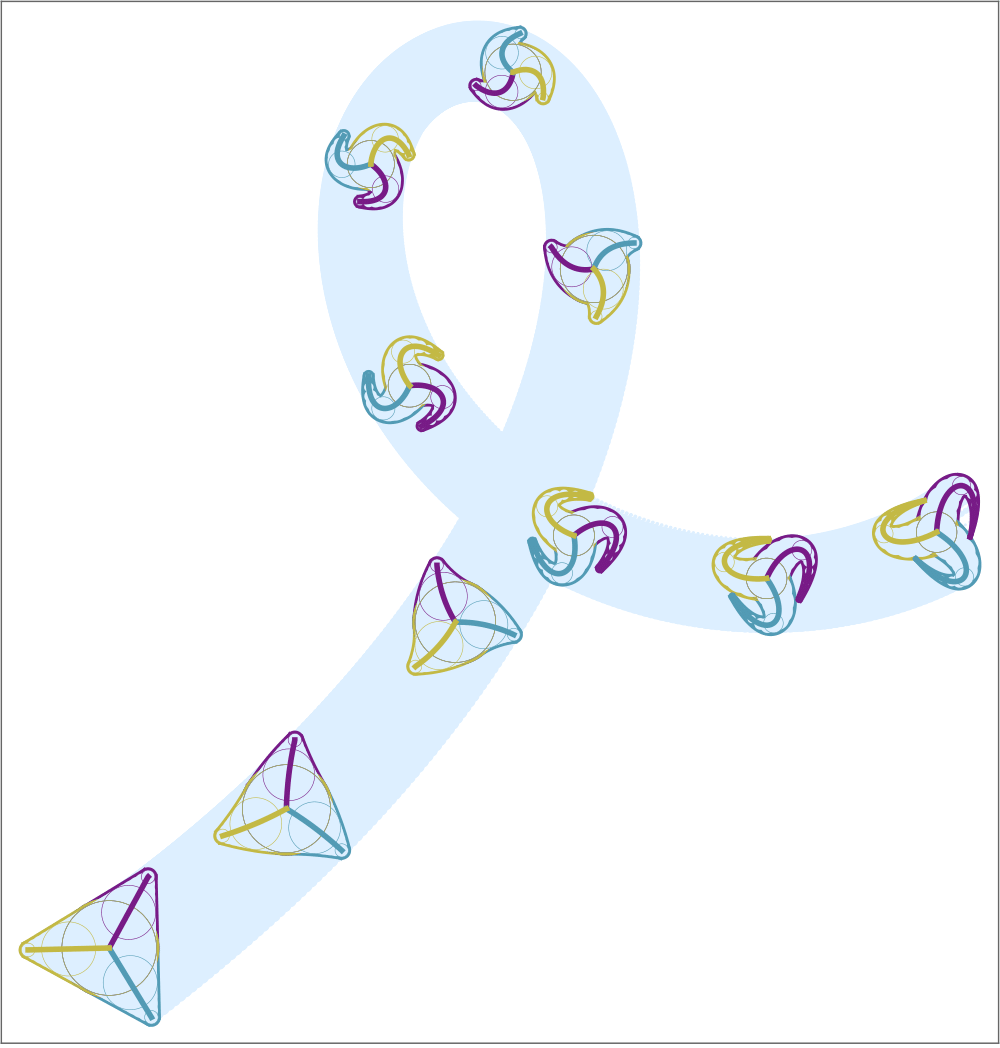}
\caption{Evolving planar domain $\Omega(t)$. At every $t \in [0,1]$ it consists of three worms.}
\label{fig_ex2domain}
\end{subfigure}%
\begin{subfigure}[t]{0.5\textwidth}
  \centering
  \includegraphics[width=0.6\linewidth]{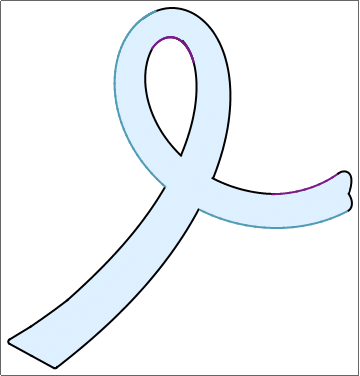}
\caption{The outer envelope of $\Omega(t)$.} 
\label{fig_ex2env}
\end{subfigure}
\end{figure}

We approximate every evolving branch separately. For each $i = 1,2,3,$ we identify the boundary and the singular curves on the surface
\begin{equation*}
    B_i(u,t), \quad (u,t) \in [0,1]^2 \ .
\end{equation*}
and approximate the corresponding worms using the method \textbf{IBI}. The cyclographic images of the singular curves are necessary to approximate the envelope of the evolving domain, approximately 32 \% of the envelope arcs are contributed by the singular curves. Figure \ref{fig_ex2superset} shows the superset of the envelope with highlighted cyclograpic images of the singular curves and Figure \ref{fig_ex2env} depicts the extracted outer envelope, where the arcs contributed by the singular curves are highlighted. 

\begin{figure}[t]
\ContinuedFloat
\centering
\begin{subfigure}[t]{\textwidth}
  \centering
  \includegraphics[width=0.4\linewidth]{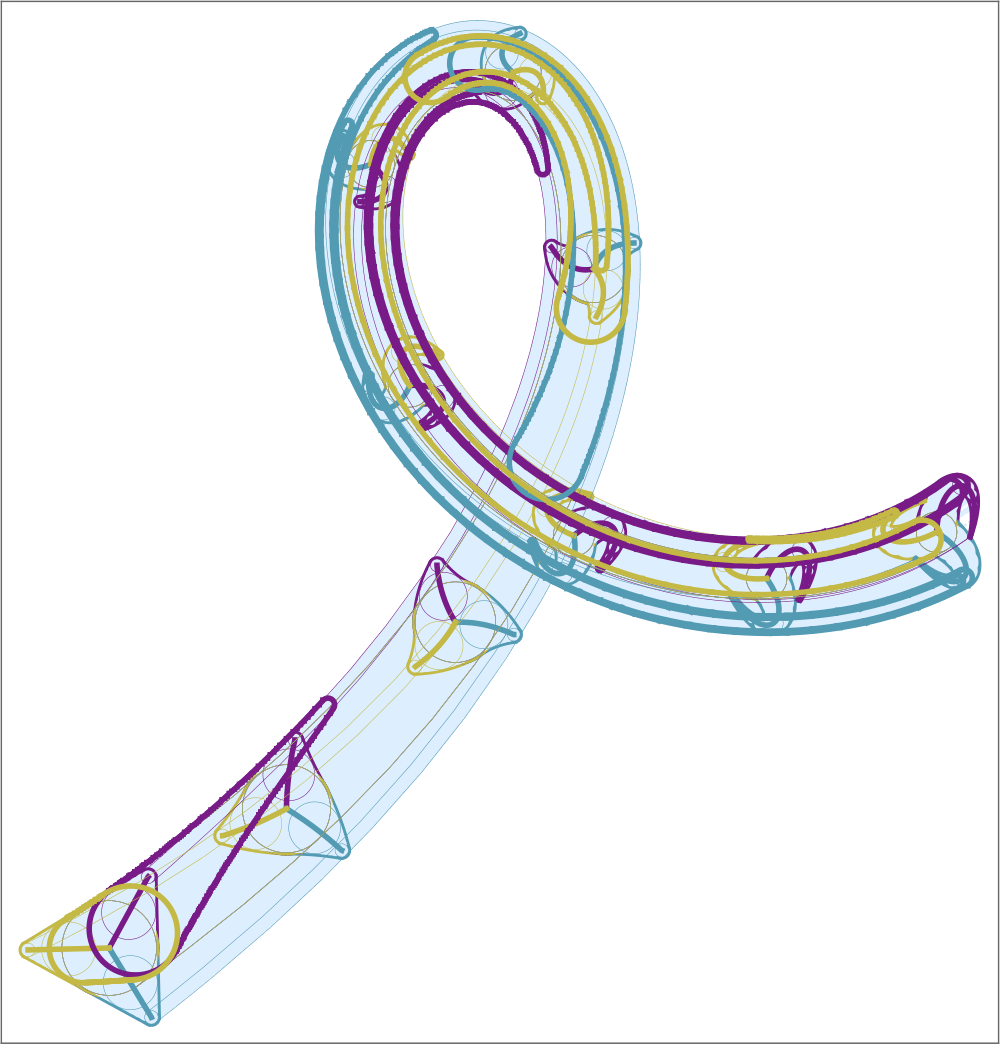}
\caption{The superset of the envelope consists of boundaries of cyclographic images of both the singular (thick) and the boundary curves.}
\label{fig_ex2superset}
\end{subfigure}%
\caption{Example 2.}
\end{figure}
\end{example}

\begin{example}
    The presented method allows us to compute envelopes of evolving domains, where the number of components changes. The domain in Figure \ref{fig_ex3domain} is composed of two worms at each instance of the time parameter. At the beginning, the worms are connected, and as the domain evolves, they separate, creating two components. We identify the boundary and the singular curves on both surfaces in the Minkowski space and again use the method \textbf{IBI} to approximate the envelope. Figure \ref{fig_ex3boundary} shows the boundaries of the cyclographic images of the boundary curves of the both worms. To obtain the whole envelope, we also need to consider the singular curves. Figure \ref{fig_ex3env} shows the extracted envelope of the evolving domain, where the boundaries of the cyclographic images of the singular curves of both worms are highlighted. In this example, approximately 14 \% of the envelope arcs are contributed by the singular curves.
\end{example}

\begin{figure}[h]
\centering
\begin{subfigure}{\textwidth}
    \includegraphics[width=\textwidth]{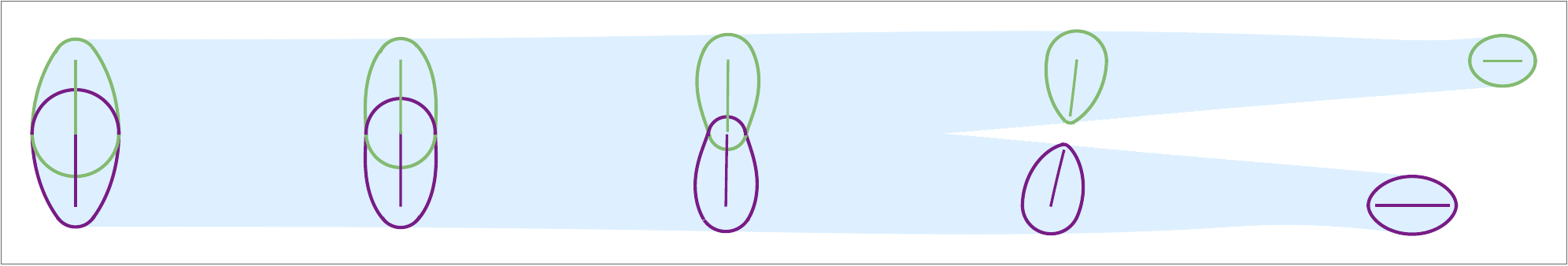}
    \caption{Evolving domain composed of two worms. At five instances of the time parameter, the worms and their MA are depicted.}
    \label{fig_ex3domain}
\end{subfigure}
\end{figure}

\begin{figure}[h]
\ContinuedFloat
\begin{subfigure}{\textwidth}
    \includegraphics[width=\textwidth]{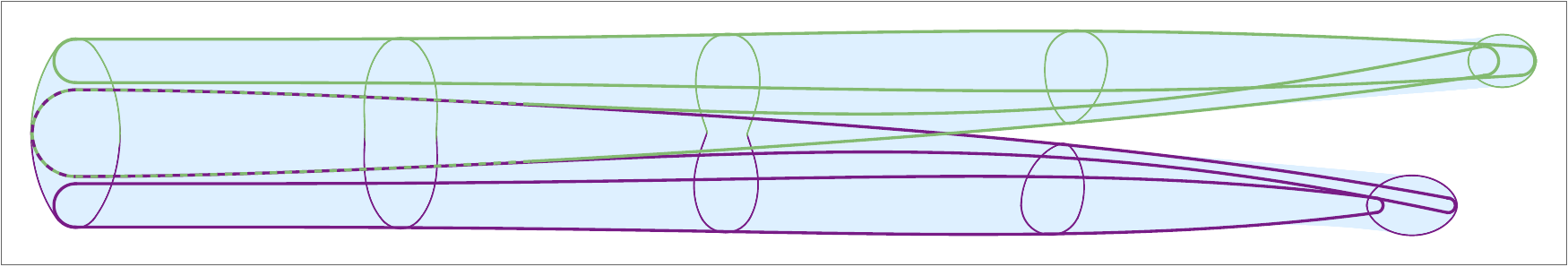}
    \caption{Boundaries of the cyclographic images of the boundary curves.}
    \label{fig_ex3boundary}
\end{subfigure}
\begin{subfigure}{\textwidth}
    \includegraphics[width=\textwidth]{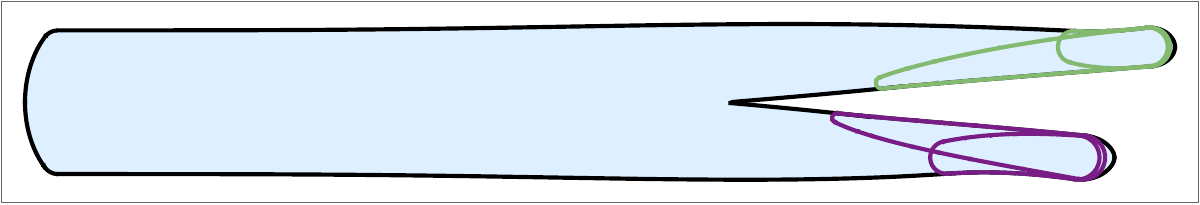}
    \caption{Extracted outer envelope (black) and the boundaries of the cyclographic images of the singular curves.}
    \label{fig_ex3env}
\end{subfigure}
        
\caption{Example 3.}
\label{fig_ex3}

\end{figure}

\section{Conclusion}\label{sec:1523}
We presented a procedure for approximating the envelope of planar domains that evolve in time. First, we considered worms, which are cyclographic images of single curve segments in $\mathbb R^{2,1}$. We characterized the envelopes of evolving worms as subsets of cyclographic images of finitely many curves in Minkowski space. We then proposed two pairs of methods to approximate these curves and their cyclographic images by arc splines. The methods exploit geometric flexibility and computational simplicity of circular arcs to efficiently and precisely approximate the envelopes.
We considered free--form domains to be unions of worms. Finally, the envelope of an evolving (free--form) domain was obtained by trimming the redundant curves efficiently using a line sweep algorithm that again utilized the properties of circular arcs.

\section*{Appendix: Points with light--like tangents}\label{appendix}
In Section \ref{sec_interpolationmethods}, we presented four methods for approximation curves in the Minkowski space $\mathbb R^{2,1}$ and their cyclographic images. We only considered curves with space--like tangents, however, light--like tangents may be present at endpoints. Figure \ref{fig:lightlike} shows such a situation. 

\begin{figure}[h!]
  \centering
  \includegraphics[width=0.4\textwidth]{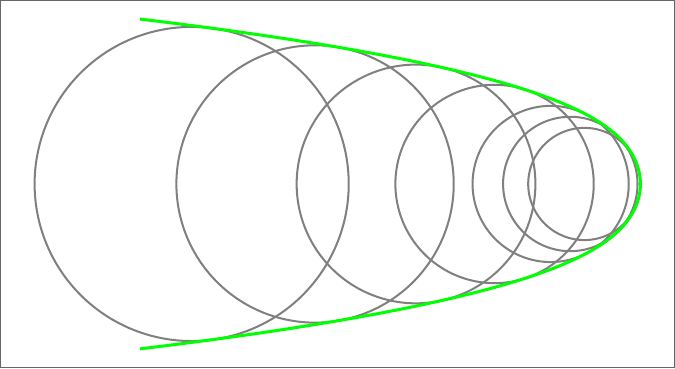}
  \caption{Minkowski curve segment with a light-like tangent at an endpoint.}
  \label{fig:lightlike}
\end{figure}

We discuss the effect of these points for the four methods:
\begin{itemize}
  \item
\textbf{DAI}: The method can be applied without any problems. The arcs of the circles that correspond to the endpoints must be added to the envelope.
  \item
\textbf{DBI}: The method cannot be used if one of the segment end points has a light-like tangent. Let $v_N = 1$ be the last parameter value, such that the tangent at the point $C_N = C(v_N)$ is light--like, i.e., $\|C(v_N)'\|_M = 0$. We sample one additional point on the curve $C(v)$,
\begin{equation*}
  C_{N-\frac 1 2} = C(\frac{v_{N-1}+v_N}{2}) \ ,
\end{equation*}
and we use a single Minkowski arc that interpolates the points 
$C_{N-1},C_{N-\frac 1 2},C_N$ to represent the last segment. See Figure \ref{fig:lastSemgent} for an illsutration.
  \item
\textbf{IAI}: The method can be used, however a reparameterization of the boundary curve is needed since the parametric speed tends to infinity as we approach the point with a light--like tangent, which corresponds to a vertex of the boundary curve.
  \item
\textbf{IBI}: The method can also be used, again after performing a reparameterization of the boundary. 
\end{itemize}

\begin{figure}[h]
  \centering
  \includegraphics[width=0.5\textwidth]{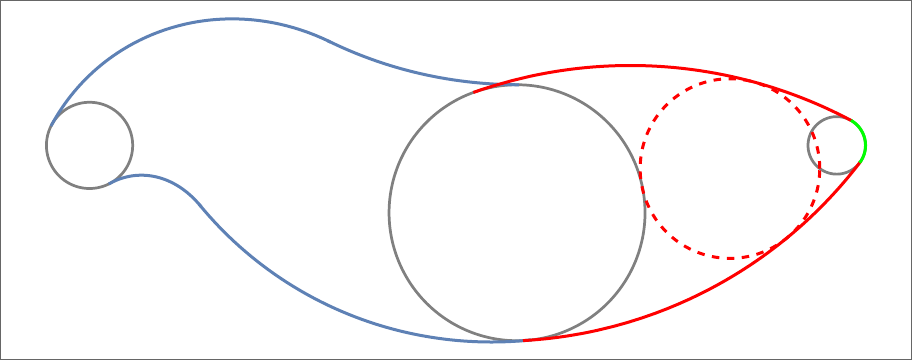}
  \caption{Handling time--like tangents in the \textbf{DBI} method. The last segment is represented by a single Minkowski arc. We sample one additional point on the curve to compute the Minkowski arc (red dashed circle).}
  \label{fig:lastSemgent}
\end{figure}

\subsection*{Acknowledgment}
This research was funded in whole or in part by the Austrian Science Fund (FWF) 10.55776/I5270.

\bibliography{references}

\end{document}